\documentclass[
    aps,
    prxquantum,
    reprint,
    superscriptaddress
]{revtex4-2}

\usepackage[T1]{fontenc}
\usepackage{lmodern}
\usepackage{setspace}

\usepackage{amsmath,amssymb}
\usepackage{amsthm}
\usepackage{bm}

\usepackage{booktabs}
\usepackage{dcolumn}
\usepackage[colorlinks=true,allcolors=blue]{hyperref}
\usepackage{enumitem}
\usepackage{xcolor}
\usepackage{cancel}

\usepackage[pdftex]{graphicx}
\usepackage{comment}
\usepackage{float}

\newtheorem{theorem}{Theorem}
\newtheorem{lemma}[theorem]{Lemma}

\theoremstyle{definition}
\newtheorem{definition}[theorem]{Definition}

\newtheorem{example}[theorem]{Example}
\newtheorem{remark}[theorem]{Remark}

\newcommand{\ket}[1]{{| {#1} \rangle}}

\begin{document}
\title{Cyclotomic Structure and Exact Recurrences in Floquet Many-Body Quantum Dynamics}

\author{Amit Anand}
\email{a63anand@uwaterloo.ca}
\affiliation{Institute for Quantum Computing, University of Waterloo, Waterloo, Ontario, Canada N2L 3G1}
\affiliation{Department of Physics and Astronomy, University of Waterloo, Waterloo, Ontario, Canada N2L 3G1}

\author{Dinesh Valluri} 
\affiliation{Institute for Quantum Computing, University of Waterloo, Waterloo, Ontario, Canada N2L 3G1}

\author{Shohini Ghose}
\affiliation{Institute for Quantum Computing, University of Waterloo, Waterloo, Ontario, Canada N2L 3G1}
\affiliation{Department of Physics and Astronomy, University of Waterloo, Waterloo, Ontario, Canada N2L 3G1}
\affiliation{Department of Physics and Computer Science, Wilfrid Laurier University, Waterloo, Ontario, Canada N2L 3C5}
\affiliation{Perimeter Institute for Theoretical Physics, 31 Caroline St N, Waterloo, Ontario, Canada N2L 2Y5}


\begin{abstract}

Periodically driven quantum systems exhibit a wide range of complex dynamics, including chaos. Yet, at special parameter values, their evolution can recur exactly. Several quantum protocols have recently been proposed that exploit exact Floquet recurrences to demonstrate genuine quantum advantage. Such recurrences can be studied through the arithmetic structure of the Floquet unitary. In this work, we establish a necessary and sufficient arithmetic criterion for exact recurrence and determine the minimal recurrence period when one exists. Because recurrence is a global property of the unitary, it holds for every initial state. Applying our framework requires only knowledge of the number field associated with the Hamiltonian parameters. We demonstrate the method using the kicked Ising chain, establishing exact recurrences in both Clifford and non-Clifford regimes. We further show that, in the recurrent non-Clifford regime, magic—a resource for quantum computation—is generated at intermediate times and undergoes exact periodic revivals enforced by the recurrence of the full Floquet unitary. By connecting number theory, quantum many-body dynamics, and quantum information, our work shows that exact recurrence extends beyond Clifford circuits and opens new avenues for magic generation and quantum hardware benchmarking.

\end{abstract}

\maketitle

\section{Introduction}

Periodically driven non-equilibrium quantum systems can generate complex many-body dynamics~\cite{Shirley1965}. Yet at special parameter values their stroboscopic evolution may close exactly after a finite number of drive cycles ~\cite{anand_davis_ghose_2024}. Such recurrences therefore define exceptional points in parameter space where the entire quasienergy spectrum is commensurate~\cite{anand2026arithmetic}, sharply contrasting with generic Floquet dynamics~\cite{dalessio_rigol_2014}, while still permitting complex intermediate-time behaviour~\cite{sommers2023crystalline}. Without access to the full spectrum of the one-period evolution operator, determining whether such a recurrence exists is difficult. A direct search over successive drive cycles can also be computationally expensive and does not provide a clear way to rule out recurrence. Because stroboscopic dynamics is generated by repeated applications of a single one-period operator, the recurrence problem acquires an arithmetic structure when that operator is defined over a number field~\cite{anand2026arithmetic}. Understanding this structure is therefore important for explaining why exact closure occurs only at special parameter values and what distinguishes these points from generic Floquet dynamics.

Over the past several decades, periodically driven non-equilibrium quantum systems have emerged as a major setting for investigating quantum dynamics. They have been studied from diverse perspectives, including quantum chaos and dynamical localization~\cite{Izrailev_Shepelyanskii_1980,fishman_rotor_anderson_1982,wang_ghose_entg_2004,Fortin_Lombardi_2019}, thermalization~\cite{kaufman_2016_quantum,bukov_many_body_2021}, and quantum--classical correspondence 
~\cite{Kumari_Ghose_vicinity_2018}. Periodic driving also supports a broad range of many-body phenomena, including quantum many-body scars~\cite{Pai_and_Pretko_scar_2019,Sen_scar_2020,Norio_scars_2020}, discrete time crystals and Floquet topological phases~\cite{roy_2017_floquet,sacha_2017_time,khemani_2019_a}.

Floquet theory provides the natural framework for studying these systems~\cite{gfloquet_1883}. For a time-periodic Hamiltonian satisfying $H(t+T)=H(t),$ the evolution over one drive period is described by
\begin{equation}
    U_F= \mathcal{T} \exp\!\left[ -i\int_{t_0}^{t_0+T} H(t)\,dt\right],
\end{equation}
where \(\mathcal{T}\) denotes time ordering and we set \(\hbar=1\)~\cite{grifoni_driven_1998}. The stroboscopic evolution after \(n\) periods is then generated by $n$ applications of the same one-time period unitary, \(U_F^n\). The ability to implement and repeatedly apply a fixed driving cycle makes Floquet systems particularly attractive for controlled experiments, and substantial progress has been made in realizing Floquet dynamics across a variety of quantum platforms~\cite{Chaudhury_2009,Neill_2016, kanem_2007_observation,Krithika_NMR_QKT_2019,Lignier2007,Jotzu2014,Mi2022TimeCrystal,Zhang2022FloquetSPT,Kim2023Ising,liao2023simulation,tindall2024efficient}.

The one-period Floquet operator incorporates the explicit time dependence over a complete drive cycle, reducing the evolution at integer multiples of the driving period to repeated applications of a fixed unitary operator~\cite{hogg_1982_recurrence,Izrailev_Shepelyanskii_1980,pandit_2022_bounds,oszmaniec_2024_saturation,anand_davis_ghose_2024,anand2026arithmetic}. Representing the stroboscopic dynamics through successive powers of this fixed operator provides a direct way to analyze evolution over many drive cycles, with quantum recurrence emerging as a natural question. Approximate quantum recurrence is the quantum analogue of the Poincaré recurrence theorem, which has long been established in classical dynamics~\cite{poincar_1890_avantpropos}. In the quantum realm, recurrence is understood in an approximate sense: under unitary evolution, a quantum state eventually returns arbitrarily close to its initial state~\cite{bocchieri_1957_quantum,hogg_1982_recurrence,Chaitanya2026recurrence}. Although this state-dependent quantum recurrence theorem is commonly formulated for systems governed by time-independent Hamiltonians, periodically driven systems admit an analogous stroboscopic description. The timescales of these approximate recurrences depend strongly on the dynamical structure of the system, typically scaling exponentially with system size in integrable many-body systems and doubly exponentially in generic non-integrable systems~\cite{gimeno_upper_2017,suskind_complexity_2018,freedman_quantumrecurrent_2024,levesque_rec_time_2025,ropotenko_the_2025}.

Beyond approximate returns of individual states, periodically driven systems can exhibit exact recurrence at special parameter values ~\cite{anand_davis_ghose_2024,bhosle_qkt__jan_2024,Purohit2026DoubleKickedTop, vonKeyserlingk2018Operator}. At these points, every initial state returns after the same finite number of drive cycles, up to a global phase. The evolution between returns can still generate entanglement or scramble information, while the exact revival provides a predictable time to recover the initial state. These features have motivated Floquet proposals for quantum metrology ~\cite{zou_2025_enhancing,biswas_2025_the} and secure quantum communication and key distribution~\cite{SreeramSanthanam2026}. Related work has used the analytically tractable dynamics of a kicked-Ising circuit to design a quantum battery~\cite{Romero2026Battery}.

 At an exact recurrence, the Floquet unitary has finite period,
\begin{equation}
    U_F^n=\tau I, \qquad \tau\in U(1).
\end{equation}
Every initial state then returns up to the same global phase, and every observable returns exactly. This state-independent condition is stronger than the approximate return of a particular state. In our previous work~\cite{anand2026arithmetic}, we showed that if the entries of a \(d\)-dimensional Floquet unitary lie in a number field \(K\), its minimal recurrence period, $n$ must satisfy
\begin{equation}
    \varphi(n)\mid [K:\mathbb{Q}]d!.
\end{equation}
where $\varphi(n)$ is an Eulier-totient function. This restricts possible periods to a finite set, but that set can be large, making separate verification of each candidate demanding.

Building on the arithmetic framework introduced in ~\cite{anand2026arithmetic}, we develop a prime-resolved procedure that determines whether a finite-dimensional Floquet unitary has an exact recurrence up to a global phase. The resulting criterion is necessary and sufficient and applies regardless of whether the dynamics are integrable or nonintegrable, or whether the unitary is Clifford or non-Clifford. We assume that the entries of the one-period unitary lie in a number field. This includes suitable parameter regimes of models such as the quantum kicked top~\cite{haake_1987}, kicked Ising chains\cite{ho2022exact}, and kicked central spin systems~\cite{biswas_2025_the}. Choosing drive parameters as rational multiples of \(\pi\) can yield algebraic unitary entries in these models, but it does not, by itself, guarantee exact recurrence.

The spectral condition for recurrence is simple to state. If the eigenvalues are \(\lambda_k=e^{i\theta_k}\), an exact recurrence exists if and only if every ratio \(\lambda_k/\lambda_1\) is a $n-th$ root of unity, or equivalently, every phase difference \(\theta_k-\theta_1\) is a rational multiple of \(2\pi\). Applying this condition directly is difficult when the spectrum is not known exactly, as is often the case for interacting Floquet systems. This motivates a criterion based on the arithmetic data of the one-period unitary rather than explicit eigenvalue expressions.

The central idea is to avoid searching directly over successive powers of \(U\). Instead, we use the theory of ramification of the cyclotomic field extensions to determine which prime factors and prime powers can possibly occur in a recurrence period. Exact recurrence requires every distinct eigen-phase to realign after the same number of drive cycles. 
If the minimal recurrence period is $n=\prod_{i=1}^{k} p_i^{e_i}$, ratios of these eigenphase must be $n-$th roots of unity. We give a method to find the prime factors $p_i$ and the indices $k_i$ by studying the discriminant of the characteristic polynomial of the unitary matrix. For the discriminant-based analysis to apply, the relevant polynomial must have a nonzero discriminant. We therefore replace the characteristic polynomial by its square-free part, thereby removing repeated roots arising from eigenvalue degeneracies while retaining every distinct eigenvalue. Once the admissible primes \(p_i\) and their maximal exponents \(e_i\) have been identified, we obtain a necessary and sufficient criterion for exact recurrence and determine the minimal recurrence period.

We apply this framework to the kicked Ising chain, a many-body Floquet model used to study quantum chaos, entanglement growth, and thermalization~\cite{Prosen1998Integrable,Prosen2000KickedIsing,bertini2019entanglement,akila2016particle,bertini2018exact,bertini2019entanglement,gopalakrishnan2019unitary}.
At special parameter values, its dynamics can also produce exact projected quantum-state designs~\cite{ho2022exact}. Kicked Ising
circuits have served as benchmarks for quantum processors and classical simulation methods~\cite{kim2023evidence,tindall2024efficient,begusic2024fast,liao2023simulation}, and their charging dynamics have been explored in a quantum battery ~\cite{Romero2026Battery}. The model's Clifford and non-Clifford regimes provide a useful comparison where Clifford evolution preserves stabilizer states and, on a finite qubit system, always recurs up to a global phase~\cite{vonKeyserlingk2018Operator,sommers2023crystalline}, whereas non-Clifford evolution can generate magic~\cite{tirrito2024anticoncentration,jha2026magic}. We take the model parameters to be rational multiples of \(\pi\), so that the Floquet operator satisfies our number-field assumption, and illustrate the arithmetic procedure for an explicit parameter set. We then study exact recurrence in both regimes. For a Clifford family, we find that the recurrence period varies non-monotonically with system size. We also show that exactly recurrent non-Clifford dynamics can generate magic at intermediate times and analytically establish recurrence for an infinite family of non-Clifford parameter choices. Thus, the class of exactly recurrent Floquet unitaries extends beyond the Clifford class, even though every Clifford unitary on a finite qubit system has finite order up to a global phase.

The paper is organized as follows. Section~\ref{sec:arithmetic cretrion for exact recurrence} develops the arithmetic criterion for exact recurrence. We first constrain the primes that can divide a recurrence period, then construct a universal exponent and use polynomial remainders to determine whether the Floquet unitary is periodic and, if so, find its minimal period. Section~\ref{sec:algorithm} we present an algorithm that can be used independently to find the minimal period of the unitary if it exist. Section~\ref{sec:recurrence in kicked ising model} applies this method to kicked Ising chains. We examine how recurrence periods depend on system size and parameters, compare Clifford and non-Clifford regimes, and study magic generation during recurrent evolution. Section~\ref{sec:discussions} summarizes the results and discusses what they reveal about the relationship between exact recurrence and circuit structure.

\section{An Arithmetic Criterion for Exact Recurrence of Floquet Unitaries over Number Fields} \label{sec:arithmetic cretrion for exact recurrence}

The purpose of the method is to derive an exact criterion that decides whether a given unitary matrix is periodic or non-periodic. The arithmetic information associated with \(U\) is used to construct a universal exponent \(E\) that is divisible by every possible minimal period. Periodicity is then decided by the equivalence
\begin{equation}
    U\ \text{is periodic}
    \quad\Longleftrightarrow\quad
    U^E\ \text{is a scalar matrix}.
\end{equation}
Thus, we provide a necessary and sufficient condition that proves periodicity, whereas a failure of this condition provides an exact certificate of non-periodicity.

Let \(U\) be a unitary matrix describing the time evolution generated by a time-periodic Hamiltonian \(H(t)=H(t+T)\), where \(T\) is the period. Let the eigenvalues of \(U\) be \(\lambda_1,\lambda_2,\ldots,\lambda_d\), where \(d\) is the dimension of \(U\). We assume that the entries of \(U\) lie in the number field \(K=\mathbb{Q}(\zeta_m)\), where \(m\) depends on the parameters of the Hamiltonian \(H(t)\).

Since \(U\) is unitary, it is diagonalizable. Therefore, in an eigenbasis, we may write
\begin{equation}
\label{eq:unitary-diagonal-form}
U =
\begin{pmatrix}
\lambda_1 & 0 & \cdots & 0 \\
0 & \lambda_2 & \cdots & 0 \\
\vdots & \vdots & \ddots & \vdots \\
0 & 0 & \cdots & \lambda_d
\end{pmatrix}_{d\times d}.
\end{equation}
A recurrence \(U^n=\tau I\) requires every eigenvalue of \(U\) to acquire the same phase after \(n\) periods. Hence, the ratios \(\lambda_i/\lambda_j\) must be \(n\)th roots of unity. The characteristic polynomial, $f_U(x) = \det(xI-U)$, encodes all eigenvalues and its coefficients are in the field $K=\mathbb{Q}(\zeta_m)$.
The characteristic polynomial of \(U\) can also be written as
\begin{equation}
\label{eq:char-poly-distinct-roots}
f_U(x) = (x-\lambda_1)(x-\lambda_2)\cdots(x-\lambda_d),
\end{equation}
where the eigenvalues are counted with algebraic multiplicity. The discriminant of this polynomial is
\begin{equation}
\label{eq:disc-char-poly}
\operatorname{Disc}(f_U) = \prod_{1\leq i<j\leq d}(\lambda_i-\lambda_j)^2 .
\end{equation}
The discriminant detects whether the polynomial has repeated roots:
\begin{equation}
\label{eq:disc-zero-repeated-root}
\operatorname{Disc}(f_U)=0
\end{equation}
if and only if \(\lambda_i=\lambda_j\) for some \(i\neq j\). Equivalently, \(\operatorname{Disc}(f_U)\neq 0\) if and only if \(f_U(x)\) has no repeated roots. 

However, spectral degeneracies do not provide additional information about the recurrence because identical eigenvalues accumulate identical phases. We therefore replace the characteristic polynomial by its square-free part \(f_{\mathrm{sf}}\), which contains each distinct eigenvalue exactly once. To remove repeated roots, we compute the square-free part of \(f_U(x)\), whose roots are precisely the distinct roots of \(f_U(x)\).

The derivative of \(f_U(x)\) is given by
\begin{equation}
\label{eq:derivative-symbolic}
f_U'(x) = \sum_{i=1}^d \prod_{j\neq i}(x-\lambda_j).
\end{equation}
Therefore,
\begin{equation}
\label{eq:derivative-at-root}
f_U'(\lambda_i) = \prod_{j\neq i}(\lambda_i-\lambda_j).
\end{equation}
Thus, \(f_U'(\lambda_i)=0\) if and only if \(\lambda_i\) is a repeated root of
\(f_U(x)\).

Now consider the greatest common divisor of \(f_U(x)\) and \(f_U'(x)\):
\begin{equation}
\label{eq:gcd-char-derivative}
\gcd\!\left(f_U(x),f_U'(x)\right)\in K[x].
\end{equation}
This greatest common divisor detects the repeated roots. More precisely, by the above observation, \(\lambda_i\) is a repeated root of \(f_U(x)\) if and only if it is a root of \(\gcd\!\left(f_U(x),f_U'(x)\right)\). Dividing by this repeated
part gives the square-free part of \(f_U(x)\):
\begin{equation}
\label{eq:square-free-polynomial-definition}
f_{\mathrm{sf}}(x) = \frac{f_U(x)}{\gcd\!\left(f_U(x),f_U'(x)\right)} \in K[x].
\end{equation}
Thus, the square-free polynomial has the form
\begin{equation}
\label{eq:square-free-polynomial}
f_{\mathrm{sf}}(x) = (x-\lambda_{i_1})(x-\lambda_{i_2})\cdots(x-\lambda_{i_k}),
\end{equation}
where \(\lambda_{i_1},\lambda_{i_2},\ldots,\lambda_{i_k}\) are the distinct
eigenvalues of \(U\). Let $L=K(\lambda_{i_1},\lambda_{i_2}...\lambda_{i_k})$ be the splitting field of \(f_U\) over \(K\).

\begin{remark}
The degree of the square free polynomial $f_{\mathrm{sf}}(x)$ is exactly the number of distinct roots of $f_U(x)$. If $U$ is periodic with the period $n$, then the number of distinct roots of $f_U(x)$ is at most $n$. Therefore $deg(f_{sf}) \leq n$.
\end{remark}

\begin{remark}
The splitting field of the original polynomial and the splitting field of its square-free part are the same, $L(f_U)=L(f_{\mathrm{sf}}).$
\end{remark}

\subsection{Ramification in a tower of number fields}
In our previous work~\cite{anand2026arithmetic}, we have shown that if $U$ is a unitary matrix with coefficients in a number field $K$ is periodic with period $n$, then $\mathbb{Q}(\zeta_n) \subseteq L$, where $L$ is the splitting field of $f_U$ or equivalently of $f_{\mathrm{sf}}$. The unitary we are dealing with has following properties, $L=K\bigl(\zeta_n,\tau^{1/n}\bigr)$ where $K$ is the base field of the unitary $U$. This establish the left hand side of the filed hierarchy diagram given in \ref{fig:field-inclusions}. The right hand side of the \ref{fig:field-inclusions}, contains  $K=\mathbb{Q}(\zeta_m)$ where $m$ depends on the Hamiltonian's parameters and thus is constant for a given set of parameters. 
\begin{figure}[H]
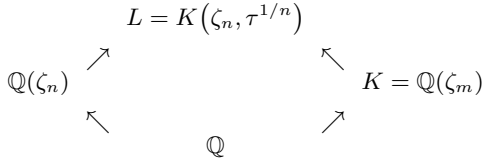

\centering
\[
\begin{array}{ccccc}
&& L=K\bigl(\zeta_n,\tau^{1/n}\bigr) && \\[1ex]
& \nearrow && \nwarrow & \\[-0.5ex]
\mathbb{Q}(\zeta_n) &&&& K=\mathbb{Q}(\zeta_m) \\[1ex]
& \nwarrow && \nearrow & \\[-0.5ex]
&& \mathbb{Q} &&
\end{array}
\]
\caption{Field inclusions associated with exact recurrence.}
\label{fig:field-inclusions}
\end{figure}

Using the relations between the different number fields, we can find all the possible prime factors of $n$ as follows. Lets assume that $n$ has a prime factorization such as $n= p_1^{e_1}....p_k^{e_k}$.  
Let \(p\) be one of the odd primes $p_1, \ldots ,p_k$. 
If the prime \(p\) divides the minimal period $n$, then it exhibits exceptional arithmetic behavior in $\mathbb{Q}(\zeta_n)$, known as \textit{ramification}, see Appendix~\ref{app:cyclotomic-discriminant} . Since the  $\mathbb{Q}(\zeta_n) \subseteq L$, all prime factors of $n$ are also ramified in $L=K\bigl(\zeta_n,\tau^{1/n}\bigr)$. 

Therefore, using the left hand side above diagram, we concluded that the rational primes dividing the period $n$ are ramified in $L$. Now, using the right hand side of the above diagram, we can classify the rational primes that are ramified in $L$ in the following ways,
\begin{lemma}\label{lem:ramification of primes}
   If \(p\) ramifies in \(L/\mathbb{Q}\), then \\ (a) either \(p\) ramifies in
\(K/\mathbb{Q}\), or \\ (b) there exists a prime ideal
\(\mathfrak{p}\subseteq\mathcal{O}_K\) lying above \(p\) that ramifies
in \(L/K\). 
\end{lemma}

The proof of this Lemma is in Appendix \ref{app:proof for ramification of primes}. The odd primes that gets ramified in the coefficient field $K=\mathbb{Q}(\zeta_m)$ divides the field discriminant \(\operatorname{Disc}(K/\mathbb{Q})\), see Appendix \ref{app:cyclotomic-discriminant}. This takes care of the rational primes that appears in $n$ from the part(a) of the Lemma~\ref{lem:ramification of primes}. The rational primes $p$, describes in part(b) of the Lemma \ref{lem:ramification of primes}, can be described in terms of the coefficients of $F_{sf}$ as detailed in Lemma \ref{lem:ramification-discriminant-cleared}.

\begin{lemma}\label{lem:ramification-discriminant-cleared} 
    If $\mathfrak{p}$ is a prime ideal of $\mathcal{O}_K$ lying above \(p\) that ramifies
in \(L\), then 
 $p$ divides $N_{K/\mathbb{Q}}(Disc(f_{\mathrm{sf}}(x))$ or $c$.
\end{lemma}

The proof of Lemma~\ref{lem:ramification-discriminant-cleared} is given in Appendix~\ref{app:ramification_discriminant_proof}. 
Since in our case $K=\mathbb{Q}(\zeta_m)$ and the characteristic polynomial $f_{\mathrm{sf}}(x)\in K[x]$, there exists an integer $c$ such that $F_{\mathrm{sf}}(x) = c f_{\mathrm{sf}}(x) \in \mathcal{O}_{K}[x]$. Evidently, $f_{\mathrm{sf}}(x)$ depends on the choice of $c$ but in our case we choose $c$ to be the smallest integer so that $F_{\mathrm{sf}}(x)$  is uniquely determined by $f_{\mathrm{sf}}(x)$. This does not change the roots or their splitting field, but it allows arithmetic results for integral polynomials to be applied. The primes dividing \(c\) must be retained separately because the polynomial has bad denominator behavior at these primes.

We can combine all the arguments above to conclude the following theorem,
\begin{theorem}
\label{thm:prime-divisors-of-cyclotomic-order}
Suppose \(U\) is a periodic unitary matrix with period \(n\) whose entries lie in a number field \(K\). Then, for every odd rational prime $p$ dividing $n$ , divides $\operatorname{Disc}(K/\mathbb{Q})$, or $c$, or $N_{K/\mathbb{Q}}\!\left(\operatorname{Disc}(F_{sf})\right)$.
\end{theorem}

\begin{remark}
When $n\equiv2\pmod4$, the prime $2$ divides $n$ but is
unramified in $\mathbb Q(\zeta_n)$. We therefore include 2 by default in the set of all possible prime divisors of $n$ among the ones determined in Theorem \ref{thm:prime-divisors-of-cyclotomic-order}.
\end{remark}

The proof of this theorem is given in Appendix \ref{app:prrof for the theorem of the prime divisors of the recurrence period}
. These quantities therefore restrict the primes that can occur in a  recurrence period. If $U$ is periodic with period $n$, we write the prime factorization of $n$ as
\begin{equation}
n=\prod_{i=1}^{k} p_i^{e_i}.
\end{equation}
The possible prime divisors $p_i$ arise from three sources. The first is the discriminant of the number field containing the entries of $U$. In our setting, $K=\mathbb{Q}(\zeta_m)$,
where $m$ is an integer depending on the entries of the unitary matrix ($U$), and hence $K$, is determined by the Hamiltonian parameters. The remaining two sources are the integer $c$ used to clear the denominators of the characteristic polynomial and the norm of the discriminant of the resulting polynomial $f_{\mathrm{sf}}(x)\in\mathcal{O}_K[x]$.
Both of these quantities also depend on the Hamiltonian parameters through the coefficients of the characteristic polynomial of \(U\).

From our previous work\cite{anand2026arithmetic}, if \(n\) is the period of \(U\), then $\varphi(n) $ divides $ [K:\mathbb{Q}]\,d!$,
where $d$ is the degree of the characteristic polynomial. In the present work, we instead consider its square-free part \(f_{\mathrm{sf}}(x)\), which has the same splitting field. Therefore, the sharper condition
\begin{equation}
\varphi(n)\mid [K:\mathbb{Q}]\,d_{\mathrm{sf}}!,
\end{equation}
holds, where $d_{\mathrm{sf}}=\deg (f_{\mathrm{sf}})$.
This divisibility condition provides an upper bound on $n$, the period of $U$.

\subsection{Universal exponent for the period}

Let \(S\) denote the finite set of allowed rational primes obtained from the discriminants and norms as given in the Theorem \ref{thm:prime-divisors-of-cyclotomic-order}.
For each \(p\in S\), we define
\begin{equation}
    a_p := \max\left\{  a\in\mathbb{Z}_{\geq 0}  : \varphi\!\left(p^a\right)\mid [K:\mathbb{Q}]\,d_{\mathrm{sf}}!\right\},
\end{equation}
and form the universal exponent
\begin{equation}
    E:=\prod_{p\in S}p^{a_p}.
\end{equation}

In Appendix~\ref{app:proof for finding maximum exponent} we provide how to find $a_p$ for each rational primes $p$. We claim that period \(n\) divides \(E\). Let \(p^a\Vert n\), meaning that \(p^a\mid n\) but \(p^{a+1}\nmid n\). We may write
\begin{equation}
    n=p^a l, \qquad \gcd(p,l)=1.
\end{equation}
By multiplicativity of Euler's totient function, $\varphi(n)=\varphi\!\left(p^a\right)\varphi(l).$ Consequently, $\varphi\!\left(p^a\right) \mid \varphi(n) \mid [K:\mathbb{Q}]\,d_{\mathrm{sf}}!.$ It follows from the definition of \(a_p\) that \(a\leq a_p\). Since this holds
for every prime-power divisor of \(n\), we obtain $ n\mid E.$ This gives the equivalence
\begin{equation}
    U\ \text{is  periodic} \quad\Longleftrightarrow\quad  U^E\ \text{is a scalar matrix}.
\end{equation}
To prove the forward implication, let \(E=qn\) for some integer $q$. Then $ U^E = \left(U^n\right)^q = \tau^q I,$
which is scalar. The reverse implication is immediate: if $ U^E=\eta I$
for some \(\eta\in\mathbb{C}\), then \(U\) is periodic with a
minimal  period dividing \(E\).
\begin{remark}
We emphasize that \(E\) need not itself satisfy $\varphi(E)\mid  [K:\mathbb{Q}]\,d_{\mathrm{sf}}!.$ The integer \(E\) is not being proposed as a candidate for the minimal period.
\end{remark}

\subsection{Testing periodicity using the minimal polynomial}

We can test whether $U^E =\tau I$ without calculating the  the matrix power \(U^E\). Since \(U\) is unitary, it is diagonalizable over
\(\mathbb{C}\). Observe that the square free polynomial, $f_{\mathrm{sf}}(x)$ given in Eq.\eqref{eq:square-free-polynomial-definition}, is precisely the minimal polynomial of the matrix $U$ \cite[Theorems~5.62 and~7.55]{Axler2024LinearAlgebra}.

\begin{theorem}\label{thm:polynomial-remainder}
Let  $U$ be a unitary, and let $f_{\mathrm{sf}}(x)=\operatorname{sf}(p_U(t))$ be the square-free part of its characteristic polynomial. Since \(U\) is diagonalizable, \(f_{\mathrm{sf}}(x)\) is the minimal polynomial of \(U\).

For \(k\geq 1\), define $\rho_k(x):=x^k\bmod f_{\mathrm{sf}}(x),$
where \(\rho_k(x)\) is the unique remainder satisfying $\deg \rho_k<\deg f_{\mathrm{sf}}=d_{\mathrm{sf}}$.
Then, for any \(\tau\in K\),
\begin{equation}
    U^k=\tau I \quad\Longleftrightarrow\quad  \rho_k(x)=\tau,
\end{equation}
where the right-hand side means that \(\rho_k(x)\) is the constant polynomial \(\tau\).
\end{theorem}

The proof of Theorem~\ref{thm:polynomial-remainder} is given in Appendix~\ref{app:polynomial-remainder-proof}.
Using the above theorem, testing the periodicity of $U$ reduces to computing $\rho_E(x) =  ^E\bmod f_{\mathrm{sf}}(x).$
If \(\rho_E(x)\) contains any nonzero term of positive degree, then \(U^E\) is not scalar and \(U\) is not periodic. If $\rho_E(x)=\eta$ is constant, then $ U^E=\eta I,$ and the period divides \(E\). This gives us the necessary and sufficient condition for $U$ being periodic. Once the universal exponent \(E\) is obtained, the minimal recurrence period is determined by testing its divisors \(n\mid E\) for which \(x^n \bmod f_{\mathrm{sf}}(x)\) is constant--equivalently, \(U^n=\tau I\)--with the polynomial remainders evaluated efficiently by binary modular exponentiation.

\section{Algorithm}\label{sec:algorithm}
Goal: Given a unitary matrix $U$ with entries in a cyclotomic field $K=\mathbb{Q}(\zeta_m)$, we want to give a criterion if $U$ is periodic or not. If it is periodic then find $n$ such that $U^n=\tau I$
\begin{itemize}
    \item Step 1: Find the characteristic polynomial $f_U(x) \in K[x]$ of $U$.\\
    \item Step 2: Find the gcd of $f_U(x)$ and $f'_U(x)$. \\
    \item Step 3: Find square free $f_{\mathrm{sf}}(x) = \frac{f_U(x)}{gcd(f_U(x),f'_U(x))}$. \\
    \item Step 4: Find the common rational integer denominator, $c$  of the coefficients of $f_{\mathrm{sf}}(x)$  such that $F_{\mathrm{sf}}(x) = c  f_{\mathrm{sf}}(x) \in \mathcal{O}_K[x]$. \\
    \item Step 5: Find the discriminant $\Delta_{sf}$ of $F_{\mathrm{sf}}(x)\in\mathcal{O}_K[x]$ . \\
    \item Step 6: Find the norm $N_{K/\mathbb{Q}}(\Delta_{sf}) \in \mathbb{Z}$.  \\ 
    \item  Step 7: Find the prime factors of the $N_{K/\mathbb{Q}}(\Delta_{sf}) $, $Disc(K/\mathbb{Q})$ and $c$. This gives us a set of primes, $\mathcal{P}=\{p_1,p_2,\dots,p_q\}$, which is guaranteed to contain all the primes appearing in $n$. \\
    \item Step 8: For each given prime $p\in\mathcal{P}$, find the maximum power, $a_p$ of $p$ such that $\varphi(p^{a_p})$ divides $d_{sf}![K:\mathbb{Q}]$.  Calculate maximum exponent $E$ as $ E=\prod_{p\in S}p^{a_p}.$ \\
    \item Step 9: Calculate  $\rho_E(x)=x^E\bmod f_{\mathrm{sf}}(x),$ if it is a not a constant then $U$ is not periodic so reject this $U$. If $\rho_E(x)$ is constant then $U$ is periodic and find $n$ using next step. 
    \item Step 10: Set $k=E$. For each $p\in\mathcal P$, while $p\mid k$, compute
   \[
    \rho_{k/p}(x)=x^{k/p}\bmod f_{\mathrm{sf}}(x).
   \]
    If $\rho_{k/p}(x)$ is constant, replace $k$ by $k/p$ and repeat
   for the same prime $p$. Otherwise, move to the next prime.
   After all primes have been tested, set $n_0=k$. This is the
   minimal period, with
  \[
     U^{n_0}=\tau I,
     \qquad
    \tau=\rho_{n_0}(x).
\]
\end{itemize}

\section{ Exact recurrence in the kicked Ising circuit}\label{sec:recurrence in kicked ising model}
In this section we will apply our method to well known model, kicked Ising 1D chain~\cite{Prosen1998Integrable,Prosen2000KickedIsing,bertini2019entanglement}.
The periodically kicked Ising model and closely related Floquet circuits have become important platforms for studying many-body quantum dynamics. At the self-dual point, particle--time duality enables exact results for random-matrix spectral statistics, entanglement growth, and reduced-state dynamics \cite{akila2016particle,bertini2018exact,bertini2019entanglement,flack2020statistics,gopalakrishnan2019unitary}, while the same structure produces exact projected quantum-state designs for suitable open-chain parameters \cite{ho2022exact}. Away from self-duality, the model provides a controlled setting for investigating quantum chaos and the transition to many-body localization \cite{braun2020transition,sierant2023stability,herrmann2023characterizing}. Special Clifford parameter choices permit efficient circuit constructions, stabilizer-state generation, cellular-automaton descriptions, and finite-size recurrence analysis \cite{mishra2015protocol,sommers2023crystalline}, whereas generic non-Clifford regimes exhibit rapid magic generation and, at particular points, exact magic and unitary revivals \cite{vonKeyserlingk2018Operator,tirrito2024anticoncentration,jha2026magic}. The model has also become a major benchmark for large-scale quantum hardware and classical simulation methods \cite{kim2023evidence,tindall2024efficient,begusic2024fast,liao2023simulation}, making it a natural testbed for developing arithmetic methods capable of identifying exact recurrences in both Clifford and non-Clifford Floquet dynamics.

We consider an open chain of \(N\) qubits evolving under the kicked-Ising Floquet unitary
\begin{equation}
U_F=U_he^{-iH_{\mathrm{Ising}}},
\label{eq:floquet_unitary}
\end{equation}
where
\begin{equation}
U_h
=
\exp\left(-ih\sum_{j=1}^{N}\sigma_j^{y}\right)
\label{eq:global_kick}
\end{equation}
is a global rotation about the \(y\) axis, and
\begin{equation}
H_{\mathrm{Ising}}
=
J\sum_{j=1}^{N-1}\sigma_j^{z}\sigma_{j+1}^{z}
+
g\sum_{j=1}^{N}\sigma_j^{z}
+
b_1\sigma_1^{z}
+
b_N\sigma_N^{z}.
\label{eq:ising_hamiltonian}
\end{equation}
Here, \(\sigma_j^\alpha\), with \(\alpha\in\{x,y,z\}\), denotes the Pauli operator acting on qubit \(j\). The parameter \(J\) determines the nearest-neighbor Ising interaction strength, \(g\) is a spatially uniform longitudinal field, \(h\) is the transverse kick angle, and \(b_1\) and \(b_N\) are additional longitudinal fields applied to the two boundary qubits. We set the duration of each Floquet period to unity, so that the rightmost factor in Eq.~\eqref{eq:floquet_unitary} describes the Ising evolution followed by the global \(y\)-rotation. For generic parameter values, this model is non-integrable and has no nontrivial global conservation laws.

\begin{example}
To illustrate the method, consider the four-qubit (\(N=4\)) kicked Ising model defined above, with $J=h=g=b_1=b_N=\frac{\pi}{4}.$
For this parameter choice, a suitable number field is $K=\mathbb{Q}(\zeta_8)$. The Hilbert-space dimension is $d=2^4=16$, and the characteristic polynomial of the Floquet unitary $U_F$ is
\begin{equation}
\begin{aligned}
f_{U_F}(x) ={}&x^{16} +\zeta_8 x^{15} +\zeta_8 x^{11} +\zeta_4 x^{10} -\zeta_4 x^6 \\
&-\zeta_8^3 x^5 -\zeta_8^3 x +1,
\end{aligned}
\end{equation}
where $\zeta_l =e^{2\pi i/l}$ denotes a primitive \(\ell\)th root of unity.
Differentiating \(f_{U_F}(x)\) gives
\begin{equation}
\begin{aligned}
f_{U_F}'(x) ={}&16x^{15} +15\zeta_8 x^{14} +11\zeta_8 x^{10} +10\zeta_4 x^9 \\
&-6\zeta_4 x^5 -5\zeta_8^3 x^4 -\zeta_8^3.
\end{aligned}
\end{equation}
The greatest common divisor of $f_{U_F}(x)$ and $f_{U_F}'(x)$ is
\begin{equation}
\gcd\!\left(f_{U_F},f_{U_F}'\right) = x^6+\zeta_8 x^5+\zeta_8 x+\zeta_4.
\end{equation}
Therefore, the square-free part of the characteristic polynomial is
\begin{equation}
\begin{aligned}
f_{\mathrm{sf}}(x) &= \frac{f_{U_F}(x)} {\gcd\!\left(f_{U_F}(x),f_{U_F}'(x)\right)} \\
&=x^{10}-\zeta_4,
\end{aligned}
\qquad
d_{\mathrm{sf}}=10.
\end{equation}
Here we can see that the original dimension of the system was $d=16$, but the dimension of square-free polynomial reduces to $d_{sf}=10$. Because \(U_F\) is unitary and hence diagonalizable, its minimal
polynomial contains each distinct eigenvalue exactly once. Since \(\zeta_4\in\mathcal{O}_K\), no denominator clearing is required as $c=1,$ implies $ f_{\mathrm{sf}}(x)=F_{\mathrm{sf}}(x).$
The polynomial discriminant is
\begin{equation}
\operatorname{Disc}\!\left(F_{\mathrm{sf}}\right) = 10^{10}\zeta_4.
\end{equation}
For \(K=\mathbb{Q}(\zeta_8)\), we have
\begin{equation}
\begin{aligned}
\left| N_{K/\mathbb{Q}} \!\left( \operatorname{Disc}\!\left(F_{\mathrm{sf}}\right) \right) \right|
&=10^{40}, \qquad \operatorname{Disc}(K/\mathbb{Q})=256,\\
N_{K/\mathbb{Q}}(c)&=1.
\end{aligned}
\end{equation}
The resulting set of rational primes is therefore $\mathcal{P}=\{2,5\}.$
The degree bound entering the arithmetic criterion is $[K:\mathbb{Q}]\,d_{\mathrm{sf}}! = 4(10)! = 14\,515\,200.$ The maximal prime-power exponents compatible with the bound are $a_2=11,$ and $a_5=3,$ which give the universal exponent $E=2^{11}5^3=256\,000.$

We first determine whether an exact recurrence exists by evaluating the remainder at the universal exponent, $\rho_E(x) = x^E \bmod f_{\mathrm{sf}}(x).$
For the present example, \(E=256000\) and \(f_{\mathrm{sf}}(x)=x^{10}-\zeta_4\). Therefore,
\begin{equation}
\begin{aligned}
\rho_E(x) &= x^{256000} \bmod \left(x^{10}-\zeta_4\right) \\
&= \left(x^{10}\right)^{25600} \bmod \left(x^{10}-\zeta_4\right) \\
&= \zeta_4^{25600} = 1.
\end{aligned}
\end{equation}
Since \(\rho_E(x)\) is constant, \(U_F^E=I\), establishing that an exact recurrence exists.

Having established recurrence, we determine the minimal recurrence period by testing the divisors \(n\mid E\) and identifying the smallest
one for which $\rho_n(x) = x^n\bmod f_{\mathrm{sf}}(x)$ is constant. We obtain
\begin{equation}
\rho_{10}(x) = x^{10}\bmod\left(x^{10}-\zeta_4\right) = \zeta_4.
\end{equation}
and for every divisor \(n\mid E\) with \(1\leq n<10\), the remainder is nonconstant. Hence the minimal recurrence period is $n_{rec}=10$ with $U_F^{10}=\zeta_4 I=iI.$

For this parameter choice, \(U_F\) is a Clifford unitary. The arithmetic criterion itself, however, makes no assumption about Cliffordness and
therefore allows exact recurrence to be studied independently of the Clifford character of the Floquet circuit.
\end{example}

\subsection{Exact recurrence in the Clifford regime}
First we restrict attention to the one-parameter family
\begin{equation}
J=h=b_1=b_N=\frac{\pi}{4}, \qquad g=\frac{r}{s}\pi, \qquad r,s\in\mathbb{Z}.
\label{eq:fixed_parameters}
\end{equation}
where $gcd(r,s)=1$ The choices \(J=h=\pi/4\) place the bulk circuit at the self-dual point, while the boundary fields \(b_1=b_N=\pi/4\) are fixed throughout the analysis. At these fixed parameter values, all gates other than the longitudinal-field rotation are Clifford. 

Since the single-qubit rotation \(e^{-ig\sigma^z}\) is Clifford precisely when \(g\) is an integer multiple of \(\pi/4\)\cite{Gottesman1998FaultTolerant}. The complete Floquet operator satisfies
\begin{equation}
U_F(g)\in\mathcal{C}_N \quad\Longleftrightarrow\quad g\in\frac{\pi}{4}\mathbb{Z},
\label{eq:clifford_condition}
\end{equation}
where \(\mathcal{C}_N\) denotes the \(N\)-qubit Clifford group.

We define the minimal recurrence period by
\begin{equation}
n_{\mathrm{rec}}(N,g)= \min\{ n\in\mathbb{N}: U_F(g)^n= \tau I_{2^N} \text{ for some }  \tau \in U(1).
\label{eq:minimal_recurrence_period}
\end{equation}
Because the Clifford group modulo global phases is finite, every Clifford Floquet operator necessarily possesses a finite recurrence period \cite{vonKeyserlingk2018Operator}. Our arithmetic method determines and verifies the corresponding minimal period. For the Clifford point $g=\frac{\pi}{4}$,
we calculated \(n_{\mathrm{rec}}\) for system sizes \(N=2,\ldots,10\). The resulting periods are displayed in Fig.~\ref{fig:recurrence_vs_qubits}. Although the recurrence period exhibits an overall tendency to increase with the number of qubits, its dependence on \(N\) is non-monotonic. In particular, the period decreases from \(48\) at \(N=5\) to \(18\) at \(N=6\), before increasing again for larger system sizes, see Fig. \ref{fig:recurrence_vs_qubits}.
\begin{figure}[t]
    \centering
    \includegraphics[
        trim={20mm 1mm 20mm 1mm},
        clip,
        width=\columnwidth
    ]{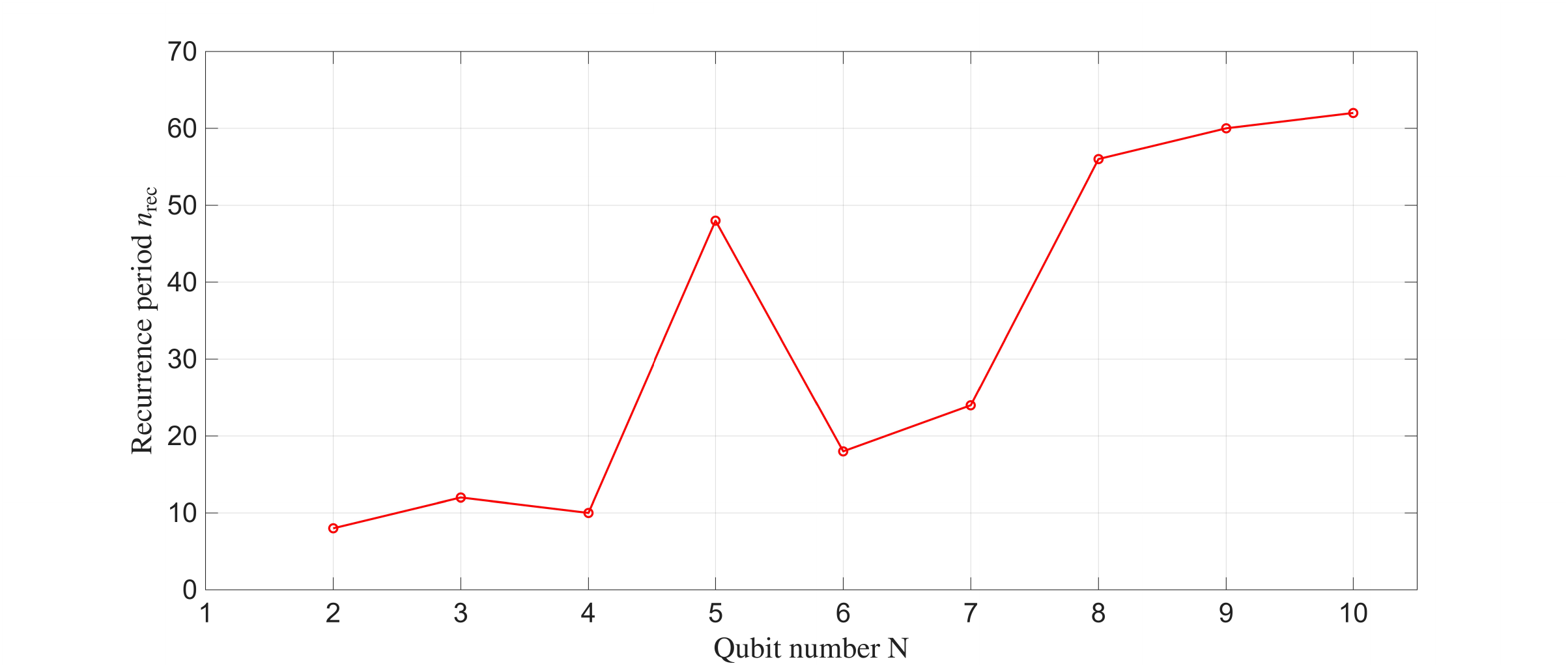}
    \caption{
        Minimal recurrence period \(T_{\mathrm{rec}}\) as a function of the number \(N\) of qubits for the kicked Ising circuit with \(J=g=h=b_1=b_N=\pi/4\). The period is defined as the smallest positive integer satisfying \(U_F^{T_{\mathrm{rec}}}=e^{i\phi}I_{2^N}\). For \(N=2,\ldots,10\), the corresponding periods are
        \(8,12,10,48,18,24,56,60,\) and \(62\), demonstrating a non-monotonic dependence on system size.
    }
    \label{fig:recurrence_vs_qubits}
\end{figure}

\subsection{Periodic and aperiodic non-Clifford regimes}
\label{sec:periodic_aperiodic_nonclifford}

We now determine whether algebraically non-Clifford kicked-Ising circuits can exhibit exact recurrence. We first retain the parameter restrictions
\begin{equation}
\label{eq:fixed_parameters_nonclifford_scan}
J=h=b_1=b_N=\frac{\pi}{4}
\end{equation}
and scan the longitudinal field over
\begin{equation}
\label{eq:nonclifford_parameter_scan}
g=\frac{r}{16}\pi, \qquad r\in\{1,2,\ldots,15\}, \qquad N\in\{3,4\}.
\end{equation}
According to the Clifford condition in
Eq.~\eqref{eq:clifford_condition}, the Floquet operator is Clifford only for $r_{\mathcal{C}_N}\in\{4,8,12\},$ for which $\frac{r}{16}\in\frac{1}{4}\mathbb{Z}.$
The remaining values, therefore generate non-Clifford Floquet circuits.

For both \(N=3\) and \(N=4\), our exact arithmetic analysis excludes a finite recurrence period at every
non-Clifford point such that $r \notin r_{\mathcal{C}_N}$. Consequently, within the parameter grid in
\eqref{eq:nonclifford_parameter_scan}, exact recurrence occurs only at the three Clifford points,
for $r_{\mathcal{C}_N}\in\{4,8,12\}$. In particular, these results show that taking all circuit angles to be rational multiples of \(\pi\) is not sufficient to guarantee an exact recurrence.

The arithmetic criterion rigorously certifies exact aperiodicity for these parameter choices. Specifically, it proves that no finite \(n\) satisfies the recurrence condition, independently of any finite search cutoff. Thus, these circuits are exactly aperiodic in the sense of Eq.~\eqref{eq:minimal_recurrence_period}. This result concerns exact equality up to a global phase and does not preclude approximate recurrences.

Non-Cliffordness does not necessarily imply exact aperiodicity. To demonstrate this, we consider a second family of parameters,
\begin{equation}
\label{eq:nonclifford_recurrent_family}
h=\frac{\pi}{2}, \qquad J=\frac{\pi}{2s}, \qquad s\in\mathbb{Z}_{>0},
\end{equation}
with longitudinal and boundary fields specified below. We first take
\begin{equation}
\label{eq:recurrent_family_fields_pi_eight}
g=b_1=b_N=\frac{\pi}{8}.
\end{equation}
At \(h=\pi/2\), the transverse kick is Clifford, so the complete Floquet unitary is Clifford if and only if its diagonal Ising-and-field layer is Clifford. A single-qubit rotation \(e^{-i\theta\sigma_z}\) is Clifford if and only if \(\theta\in\frac{\pi}{4}\mathbb{Z}\)~\cite{Gottesman1998FaultTolerant}. Define the effective on-site angles by \(\theta_1=g+b_1\), \(\theta_N=g+b_N\), and \(\theta_j=g\) for \(2\leq j\leq N-1\). For the chosen fields, the endpoint angles are \(\pi/4\), while every interior angle is \(\pi/8\). The diagonal layer commutes with every \(\sigma_z^{(j)}\), so its Clifford status can be checked by whether conjugation maps each \(\sigma_x^{(j)}\) to a Pauli string. For this
parameter family, that condition requires \(J\in\frac{\pi}{4}\mathbb{Z}\) and every \(\theta_j\in\frac{\pi}{4}\mathbb{Z}\). Consequently, for \(N\geq3\), the complete Floquet unitary is non-Clifford for every \(s\geq1\). For \(N=2\), there are no interior qubits, and the complete circuit is Clifford for \(s=1,2\) and non-Clifford for \(s\geq3\).

Our arithmetic method first detected exact recurrence throughout the parameter range
\begin{equation}
N=2,3,\ldots,12, \qquad s=1,2,\ldots,50,
\end{equation}
and returned the minimal period $n_{\mathrm{rec}}=2s$
for every examined pair \((N,s)\). The independence of the detected period from \(N\) suggested that this behavior is an exact structural property of the Floquet operator rather than a finite-size coincidence.

Guided by this arithmetic result, we subsequently proved that, for an open chain with \(N\geq2\),
\begin{equation}
\label{eq:analytical_recurrence_result}
U_F^{2s} = (-1)^{Ns+N-1}I_{2^N},
\end{equation}
and that no smaller positive power of \(U_F\) is proportional to the identity, resulting in $n_{\mathrm{rec}}=2s.$
The proof is presented in Appendix~\ref{app:analytical_recurrence_proof}. The recurrence follows from the cancellation of the longitudinal and boundary fields after two Floquet periods at \(h=\pi/2\). In particular, this result holds independently of \(g\), \(b_1\), and \(b_N\).
Eq.~\eqref{eq:analytical_recurrence_result} therefore establishes an infinite family of exactly recurrent non-Clifford Floquet circuits.

\begin{figure}[t]
    \centering
    \includegraphics[
        trim={20mm 1mm 20mm 1mm},
        clip,
        width=\columnwidth
    ]{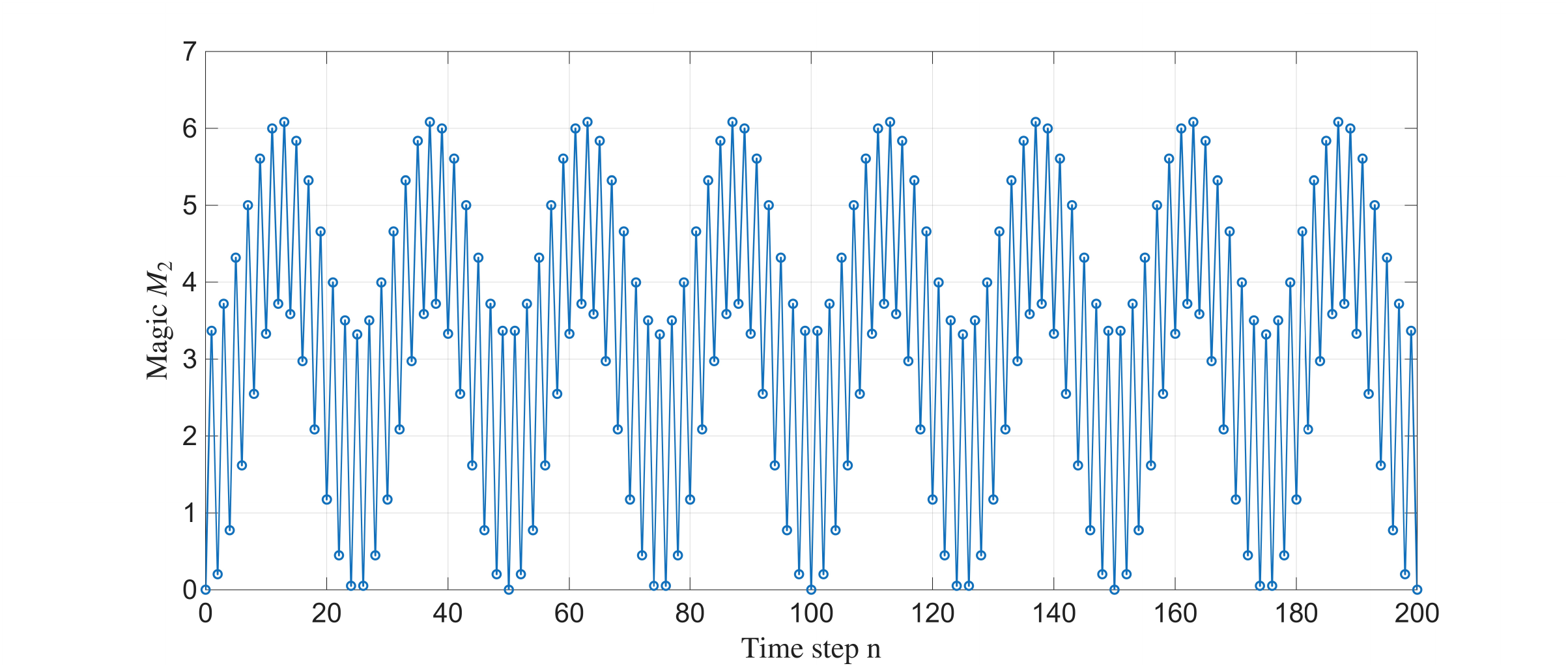}
    \caption{
        Magic generation in a periodically kicked Ising chain. The second stabilizer R\'enyi entropy \(M_2\) is shown as a function of the number \(t\) of Floquet periods for a chain of \(N=10\) qubits. The parameters are \(J=\pi/100\), \(g=\pi/8\), \(b_1=b_L=\pi/8\), \(h=\pi/2\), and \(\tau=1\). Initial state is taken as stablizer state $\psi(0)=\ket{+}^N$. Each integer time step corresponds to one application of the Floquet operator \(U_F\). For this parameter set, $U^{200}=I$
    }
    \label{fig:magic_vs_time}
\end{figure}

\subsection{Magic dynamics in the recurrent non-Clifford regime}
\label{sec:magic_nonclifford_recurrence}

We next examine whether the exactly recurrent non-Clifford circuits identified above generate magic during their intermediate-time evolution. 
We initialize the system in a fixed \(N\)-qubit stabilizer state,
\begin{equation}
\label{eq:initial_stabilizer_state}
\lvert\psi(0)\rangle = \lvert\psi_{\mathrm{stab}}\rangle = \lvert+_x\rangle^{\otimes N},
\end{equation}
 where \(\lvert+_x\rangle\) is the \(+1\) eigenstate of the Pauli
operator \(\sigma^x\), and evolve it for \(n\) Floquet periods according to
\begin{equation}
\label{eq:time_evolved_stabilizer_state}
\lvert\psi(n)\rangle = U_F^n\lvert\psi_{\mathrm{stab}}\rangle.
\end{equation}
The magic of the resulting pure state is quantified using the second stabilizer R\'enyi entropy introduced in Ref.~\cite{Leone2022Stabilizer},
\begin{equation}
\label{eq:M2_main_text}
M_2(n) = -\log_2\!\left[ \frac{1}{2^N} \sum_{P\in\mathcal{P}_N} \left| \langle\psi(n)\rvert P \lvert\psi(n)\rangle \right|^4 \right],
\end{equation}
where $\mathcal{P}_N = \{I,X,Y,Z\}^{\otimes N}$
is the set of \(4^N\) Hermitian Pauli strings, with global phases omitted. We evaluate Eq.~\eqref{eq:M2_main_text} exactly using the XOR-based fast Walsh--Hadamard-transform algorithm developed in Refs.~\cite{Huang2025XORFWHT,Xiao2026FWHT,Sierant2026Magic}. 

For a pure state, the second stabilizer R\'enyi entropy is faithful:
\begin{equation}
\label{eq:M2_faithfulness}
M_2(n)=0
\quad\Longleftrightarrow\quad
\lvert\psi(n)\rangle
\text{ is a stabilizer state},
\end{equation}
whereas $M_2(n)>0$ certifies that \(\lvert\psi(n)\rangle\) is a non-stabilizer, or magic, state \cite{Leone2022Stabilizer}. Since the initial state is a stabilizer state has $M_2(0)=0$ and Clifford unitaries preserve the set of stabilizer states, observing \(M_2(n)>0\) at any intermediate time certifies that \(U_F^n\), and consequently \(U_F\), is non-Clifford. The converse is not generally true, since a non-Clifford unitary may map a particular stabilizer input to another stabilizer state at selected times.

To calculate the magic, we therefore choose
\begin{equation}
\label{eq:magic_recurrent_parameters}
g=b_1=b_L=\frac{\pi}{8},
\qquad
h=\frac{\pi}{2},
\qquad
J=\frac{\pi}{2s}.
\end{equation}
These parameters remain within the exactly recurrent family and satisfy Eq~\eqref{eq:analytical_recurrence_result}. 
In particular, the data presented in Fig.~\ref{fig:magic_vs_time} correspond to $N=10,  s=50, n_{\mathrm{rec}}=100.$
For these parameter values given in Eq.\eqref{eq:magic_recurrent_parameters}, the Floquet operator satisfies
\begin{equation}
\label{eq:nonclifford_recurrent_conditions}
U_F\notin\mathcal{C}_N, \qquad U_F^{n_{\mathrm{rec}}}=e^{i\phi}I_{2^N},
\end{equation}
where \(\mathcal{C}_N\) denotes the \(N\)-qubit Clifford group and \(n_{\mathrm{rec}}\) is the minimal recurrence period.

The exact recurrence in Eq.~\eqref{eq:nonclifford_recurrent_conditions} implies
\begin{equation}
\begin{aligned}
\lvert\psi(n+n_{\mathrm{rec}})\rangle
&= U_F^{n+n_{\mathrm{rec}}} \lvert\psi_{\mathrm{stab}}\rangle \nonumber\\
&= e^{i\phi} U_F^n\lvert\psi_{\mathrm{stab}}\rangle = e^{i\phi}\lvert\psi(n)\rangle.
\end{aligned}
\label{eq:state_recurrence_magic}
\end{equation}
Because \(M_2\) is invariant under a global phase, it follows that
\begin{equation}
\label{eq:M2_periodicity}
M_2(n+n_{\mathrm{rec}}) = M_2(n).
\end{equation}
In particular, the initial stabilizer state is recovered at every non-negative integer multiple of the recurrence period, so that $M_2(k n_{\mathrm{rec}})=0, k\in\mathbb{Z}_{\geq0}.$
The nonzero intermediate values shown in Fig.~\ref{fig:magic_vs_time} therefore demonstrate that the circuit generates magic before returning exactly to the initial stabilizer state. 

For any pure state of \(N\) qubits, the second stabilizer R\'enyi entropy satisfies the general bound \cite{Leone2022Stabilizer}
\begin{equation}
\label{eq:M2_bound_main}
0\leq M_2(n)\leq\log_2\!\left(\frac{2^N+1}{2}\right)=N-1+\log_2\!\left(1+2^{-N}\right).
\end{equation}
For the \(N=10\) system considered in Fig.~\ref{fig:magic_vs_time}, this bound becomes
\begin{equation}
\label{eq:M2_bound_L10}
0 \leq M_2(n) \leq \log_2\!\left(\frac{1025}{2}\right) \simeq 9.0014.
\end{equation}
During the time evolution of $M_2$ shown in Fig.~\ref{fig:magic_vs_time}, the
maximum observed value is $M_2^{\mathrm{max}}\approx6.0822.$

We emphasize that the minimal period of \(M_2(n)\) may be a proper divisor of \(n_{\mathrm{rec}}\). Therefore, a revival of \(M_2\) alone does not establish the recurrence of the complete Floquet operator; the latter is determined independently through the operator identity in Eq.~\eqref{eq:nonclifford_recurrent_conditions}. Related oscillatory magic dynamics and exact revivals were recently obtained analytically for a kicked-Ising Floquet chain at the solvable parameter point \(g_{\mathrm{kick}}=\pi/2\) and \(\theta_{ZZ}=\pi/8\) \cite{jha2026magic}. At this point, the stabilizer R\'enyi entropy revives after every two kicks, while the complete Floquet operator satisfies $U_F^{8}\propto I.$
This distinction also illustrates that a revival of the magic measure need not coincide with the recurrence period of the complete Floquet operator.

\section{Summary and Discussions}\label{sec:discussions}

Previous studies of exact recurrences in periodically driven systems have proposed several quantum protocols with genuine quantum advantage~\cite{zou_2025_enhancing,biswas_2025_the,Romero2026Battery,SreeramSanthanam2026}. In this work, we propose an arithmetic criterion for exact recurrence of a Floquet unitary and establish a necessary and sufficient condition for its occurrence. We extend our previous work~\cite{anand2026arithmetic} and here we begin with the characteristic polynomial of the unitary and study the relation between its roots and the polynomial discriminant. This approach is possible because of the inclusion of the various fields involved as in Fig.\ref{fig:field-inclusions}, forced by the periodicity of the unitary.
Their hierarchy allows us to follow the behavior of primes in the the relevant field extensions and consequently identify which primes can divide the recurrence period, when it exists. Crucially, the only field information we need at the outset is the number field determined by the Hamiltonian parameters; the remaining restrictions follow from the unitary and these field relations. We then use an upper bound on the period to determine its minimal value. In many physical systems, the parameters are rational multiples of $\pi$, allowing the relevant number field to be specified explicitly. Our framework applies to finite-dimensional unitaries with entries in a number field and does not require a particular operator decomposition or circuit structure. Additional properties of a given model may permit simpler recurrence tests or more efficient algorithms. 

We apply our framework to the well-known one-dimensional kicked Ising chain, considering several sets of Hamiltonian parameters in both the Clifford and non-Clifford regimes~\cite{Prosen1998Integrable,Prosen2000KickedIsing,bertini2019entanglement}. In Ref.~\cite{sommers2023crystalline}, the authors studied recurrence times in translation-invariant dual-unitary Clifford circuits using Clifford quantum cellular automata, including a class related to self-dual kicked Ising. Their recurrence times are defined modulo Pauli signs and spatial shifts. Here, we test the condition \(U_F^n=\tau I\), including at parameter values outside the Clifford and self-dual regimes. For all Clifford parameter sets studied, the corresponding Floquet unitary has a finite recurrence period up to a global phase \cite{vonKeyserlingk2018Operator}. However, for fixed parameters, the minimal period does not necessarily increase as the number of qubits increases. We observe this non-monotonic dependence on system size across the Clifford parameter sets considered. 

In the non-Clifford regime, our framework identifies both exactly recurrent and nonrecurrent dynamics. In particular, for certain parameter values examined in Eq.~\eqref{eq:nonclifford_parameter_scan}, our arithmetic criterion proves that
\begin{equation}
\label{eq:infinite_projective_order}
U_F^n\not\propto I_{2^N}
\qquad
\text{for any }n\in\mathbb{N},
\end{equation}
so these Floquet unitaries have infinite period. We also identify a non-Clifford family that satisfies Eq.~\eqref{eq:analytical_recurrence_result} with a minimal recurrence period \(n_{\mathrm{rec}}=2s\) and therefore has finite period. For this family, the minimal period is independent of system size, in contrast with the system-size-dependent periods found for the Clifford parameter sets above. We analytically prove recurrences for this family of parameters. Taken together, these results show that non-Cliffordness alone does not determine whether exact recurrence occurs. Although every \(N\)-qubit Clifford unitary exhibits an exact recurrence up to a global phase, our results show that such recurrences are not restricted to Clifford dynamics. The non-Clifford family identified here also satisfies \(U_F^{n_{\mathrm{rec}}}=\tau I\).

We also studied magic generation during non-Clifford recurrent dynamics. Starting from stabilizer states, these circuits generate magic at intermediate times before returning to their initial states. The period of the magic
dynamics divides the recurrence period of the Floquet unitary such that magic vanishes when the state returns at \(n_{\mathrm{rec}}\). Shorter magic revivals have also been observed in Ref.~\cite{jha2026magic}, where the authors identified intermediate Clifford points at which magic vanishes. Our infinite family of recurrent, magic-generating parameter sets establishes that exact recurrence can coexist with non-Clifford evolution and revivals in magic is the consequence of the recurrence dynamics. The dynamics leave the stabilizer sector before returning, so their intermediate behavior cannot be captured by stabilizer simulation alone. These circuits therefore provide a controlled means of generating magic as a resource for quantum computation. By selecting the Hamiltonian parameters, we can tune both the amount of magic generated and the time for which it persists before revival. 

The known recurrence time also makes these dynamics promising for benchmarking quantum computers. Fischer \textit{et al.} used non-Clifford kicked Ising circuits on a 91-qubit processor as hardware benchmarks, comparing measured autocorrelation with exact predictions at the dual-unitary point~\cite{Fischer2026Chaos}. A benchmark based on recurrence requires a mathematical criterion to identify which parameter sets recur and to determine their exact periods; our arithmetic framework provides this characterization. In particular, the exact recurrence period of the non-Clifford family identified here depends on a tunable Hamiltonian parameter, allowing predicted returns at progressively later times. Once \(U_F^{n_{\mathrm{rec}}}=\tau I\) is established, the ideal evolution also returns at every multiple of \(n_{\mathrm{rec}}\). These features provide exact reference points over a range of runtimes, allowing deviations from the predicted returns to probe errors accumulated during longer circuits. Because our framework derives these predictions from the Floquet unitary, it is independent of the hardware platform and could guide analogous tests on other devices. This prospect motivates the study of exact recurrences in other driven systems. Alongside proposed uses of recurrences in Hamiltonian-parameter estimation~\cite{zou_2025_enhancing} and quantum key distribution~\cite{SreeramSanthanam2026}, these applications show that exact recurrence is more than an isolated feature of selected parameter values. It can provide a basis for quantum protocols with potential advantages across distinct settings.

\begin{acknowledgments}
This work was supported in part by the Natural Sciences and Engineering Research Council of Canada (NSERC). We acknowledge fruitful discussions with Sanchit Srivastava and Yash Totani.  Wilfrid Laurier University and the University of Waterloo are located in the traditional territory of the Neutral, Anishnawbe and Haudenosaunee peoples. We thank them for allowing us to conduct this research on their land. 

The authors used ChatGPT (Sol 5.6) used to improve clarity and presentation of the manuscript, identifying references, formulating and checking few mathematical arguments, and developing and debugging the GAP implementation. These activities were carried out under the authors’ direction. The authors take full responsibility for the content, results, and conclusions of this work.
\end{acknowledgments}



\section*{Data Availability}

Codes and data are available upon request.  


\bibliography{ref}

\begin{thebibliography}{76}%
\makeatletter
\providecommand \@ifxundefined [1]{%
 \@ifx{#1\undefined}
}%
\providecommand \@ifnum [1]{%
 \ifnum #1\expandafter \@firstoftwo
 \else \expandafter \@secondoftwo
 \fi
}%
\providecommand \@ifx [1]{%
 \ifx #1\expandafter \@firstoftwo
 \else \expandafter \@secondoftwo
 \fi
}%
\providecommand \natexlab [1]{#1}%
\providecommand \enquote  [1]{``#1''}%
\providecommand \bibnamefont  [1]{#1}%
\providecommand \bibfnamefont [1]{#1}%
\providecommand \citenamefont [1]{#1}%
\providecommand \href@noop [0]{\@secondoftwo}%
\providecommand \href [0]{\begingroup \@sanitize@url \@href}%
\providecommand \@href[1]{\@@startlink{#1}\@@href}%
\providecommand \@@href[1]{\endgroup#1\@@endlink}%
\providecommand \@sanitize@url [0]{\catcode `\\12\catcode `\$12\catcode `\&12\catcode `\#12\catcode `\^12\catcode `\_12\catcode `\%12\relax}%
\providecommand \@@startlink[1]{}%
\providecommand \@@endlink[0]{}%
\providecommand \url  [0]{\begingroup\@sanitize@url \@url }%
\providecommand \@url [1]{\endgroup\@href {#1}{\urlprefix }}%
\providecommand \urlprefix  [0]{URL }%
\providecommand \Eprint [0]{\href }%
\providecommand \doibase [0]{https://doi.org/}%
\providecommand \selectlanguage [0]{\@gobble}%
\providecommand \bibinfo  [0]{\@secondoftwo}%
\providecommand \bibfield  [0]{\@secondoftwo}%
\providecommand \translation [1]{[#1]}%
\providecommand \BibitemOpen [0]{}%
\providecommand \bibitemStop [0]{}%
\providecommand \bibitemNoStop [0]{.\EOS\space}%
\providecommand \EOS [0]{\spacefactor3000\relax}%
\providecommand \BibitemShut  [1]{\csname bibitem#1\endcsname}%
\let\auto@bib@innerbib\@empty
\bibitem [{\citenamefont {Shirley}(1965)}]{Shirley1965}%
  \BibitemOpen
  \bibfield  {author} {\bibinfo {author} {\bibfnamefont {J.~H.}\ \bibnamefont {Shirley}},\ }\bibfield  {title} {\bibinfo {title} {Solution of the schrödinger equation with a hamiltonian periodic in time},\ }\href {https://doi.org/10.1103/PhysRev.138.B979} {\bibfield  {journal} {\bibinfo  {journal} {Physical Review}\ }\textbf {\bibinfo {volume} {138}},\ \bibinfo {pages} {B979} (\bibinfo {year} {1965})}\BibitemShut {NoStop}%
\bibitem [{\citenamefont {Anand}\ \emph {et~al.}(2024)\citenamefont {Anand}, \citenamefont {Davis},\ and\ \citenamefont {Ghose}}]{anand_davis_ghose_2024}%
  \BibitemOpen
  \bibfield  {author} {\bibinfo {author} {\bibfnamefont {A.}~\bibnamefont {Anand}}, \bibinfo {author} {\bibfnamefont {J.}~\bibnamefont {Davis}},\ and\ \bibinfo {author} {\bibfnamefont {S.}~\bibnamefont {Ghose}},\ }\bibfield  {title} {\bibinfo {title} {Quantum recurrences in the kicked top},\ }\href {https://doi.org/10.1103/PhysRevResearch.6.023120} {\bibfield  {journal} {\bibinfo  {journal} {Physical Review Research}\ }\textbf {\bibinfo {volume} {6}},\ \bibinfo {pages} {023120} (\bibinfo {year} {2024})}\BibitemShut {NoStop}%
\bibitem [{\citenamefont {Anand}\ \emph {et~al.}(2026)\citenamefont {Anand}, \citenamefont {Valluri}, \citenamefont {Davis},\ and\ \citenamefont {Ghose}}]{anand2026arithmetic}%
  \BibitemOpen
  \bibfield  {author} {\bibinfo {author} {\bibfnamefont {A.}~\bibnamefont {Anand}}, \bibinfo {author} {\bibfnamefont {D.}~\bibnamefont {Valluri}}, \bibinfo {author} {\bibfnamefont {J.}~\bibnamefont {Davis}},\ and\ \bibinfo {author} {\bibfnamefont {S.}~\bibnamefont {Ghose}},\ }\bibfield  {title} {\bibinfo {title} {Quantum recurrences and the arithmetic of {F}loquet dynamics},\ }\href {https://doi.org/10.22331/q-2026-04-20-2074} {\bibfield  {journal} {\bibinfo  {journal} {Quantum}\ }\textbf {\bibinfo {volume} {10}},\ \bibinfo {pages} {2074} (\bibinfo {year} {2026})}\BibitemShut {NoStop}%
\bibitem [{\citenamefont {D'Alessio}\ and\ \citenamefont {Rigol}(2014)}]{dalessio_rigol_2014}%
  \BibitemOpen
  \bibfield  {author} {\bibinfo {author} {\bibfnamefont {L.}~\bibnamefont {D'Alessio}}\ and\ \bibinfo {author} {\bibfnamefont {M.}~\bibnamefont {Rigol}},\ }\bibfield  {title} {\bibinfo {title} {Long-time behavior of isolated periodically driven interacting lattice systems},\ }\href {https://doi.org/10.1103/PhysRevX.4.041048} {\bibfield  {journal} {\bibinfo  {journal} {Physical Review X}\ }\textbf {\bibinfo {volume} {4}},\ \bibinfo {pages} {041048} (\bibinfo {year} {2014})}\BibitemShut {NoStop}%
\bibitem [{\citenamefont {Sommers}\ \emph {et~al.}(2023)\citenamefont {Sommers}, \citenamefont {Huse},\ and\ \citenamefont {Gullans}}]{sommers2023crystalline}%
  \BibitemOpen
  \bibfield  {author} {\bibinfo {author} {\bibfnamefont {G.~M.}\ \bibnamefont {Sommers}}, \bibinfo {author} {\bibfnamefont {D.~A.}\ \bibnamefont {Huse}},\ and\ \bibinfo {author} {\bibfnamefont {M.~J.}\ \bibnamefont {Gullans}},\ }\bibfield  {title} {\bibinfo {title} {Crystalline quantum circuits},\ }\href {https://doi.org/10.1103/PRXQuantum.4.030313} {\bibfield  {journal} {\bibinfo  {journal} {PRX Quantum}\ }\textbf {\bibinfo {volume} {4}},\ \bibinfo {pages} {030313} (\bibinfo {year} {2023})}\BibitemShut {NoStop}%
\bibitem [{\citenamefont {Izrailev}\ and\ \citenamefont {Shepelyanskii}(1980)}]{Izrailev_Shepelyanskii_1980}%
  \BibitemOpen
  \bibfield  {author} {\bibinfo {author} {\bibfnamefont {F.~M.}\ \bibnamefont {Izrailev}}\ and\ \bibinfo {author} {\bibfnamefont {D.~L.}\ \bibnamefont {Shepelyanskii}},\ }\bibfield  {title} {\bibinfo {title} {Quantum resonance for a rotator in a nonlinear periodic field},\ }\href {https://doi.org/10.1007/BF01029131} {\bibfield  {journal} {\bibinfo  {journal} {Theoretical and Mathematical Physics}\ }\textbf {\bibinfo {volume} {43}},\ \bibinfo {pages} {553–561} (\bibinfo {year} {1980})}\BibitemShut {NoStop}%
\bibitem [{\citenamefont {Fishman}\ \emph {et~al.}(1982)\citenamefont {Fishman}, \citenamefont {Grempel},\ and\ \citenamefont {Prange}}]{fishman_rotor_anderson_1982}%
  \BibitemOpen
  \bibfield  {author} {\bibinfo {author} {\bibfnamefont {S.}~\bibnamefont {Fishman}}, \bibinfo {author} {\bibfnamefont {D.~R.}\ \bibnamefont {Grempel}},\ and\ \bibinfo {author} {\bibfnamefont {R.~E.}\ \bibnamefont {Prange}},\ }\bibfield  {title} {\bibinfo {title} {Chaos, quantum recurrences, and anderson localization},\ }\href {https://doi.org/10.1103/PhysRevLett.49.509} {\bibfield  {journal} {\bibinfo  {journal} {Physical Review Letter}\ }\textbf {\bibinfo {volume} {49}},\ \bibinfo {pages} {509} (\bibinfo {year} {1982})}\BibitemShut {NoStop}%
\bibitem [{\citenamefont {Wang}\ \emph {et~al.}(2004)\citenamefont {Wang}, \citenamefont {Ghose}, \citenamefont {Sanders},\ and\ \citenamefont {Hu}}]{wang_ghose_entg_2004}%
  \BibitemOpen
  \bibfield  {author} {\bibinfo {author} {\bibfnamefont {X.}~\bibnamefont {Wang}}, \bibinfo {author} {\bibfnamefont {S.}~\bibnamefont {Ghose}}, \bibinfo {author} {\bibfnamefont {B.~C.}\ \bibnamefont {Sanders}},\ and\ \bibinfo {author} {\bibfnamefont {B.}~\bibnamefont {Hu}},\ }\bibfield  {title} {\bibinfo {title} {Entanglement as a signature of quantum chaos},\ }\href {https://doi.org/10.1103/PhysRevE.70.016217} {\bibfield  {journal} {\bibinfo  {journal} {Phys. Rev. E}\ }\textbf {\bibinfo {volume} {70}},\ \bibinfo {pages} {016217} (\bibinfo {year} {2004})}\BibitemShut {NoStop}%
\bibitem [{\citenamefont {Fortin}\ and\ \citenamefont {Lombardi}(2019)}]{Fortin_Lombardi_2019}%
  \BibitemOpen
  \bibfield  {author} {\bibinfo {author} {\bibfnamefont {S.}~\bibnamefont {Fortin}}\ and\ \bibinfo {author} {\bibfnamefont {O.}~\bibnamefont {Lombardi}},\ }\bibfield  {title} {\bibinfo {title} {The correspondence principle and the understanding of decoherence},\ }\href {https://doi.org/10.1007/s10701-019-00309-4} {\bibfield  {journal} {\bibinfo  {journal} {Foundations of Physics}\ }\textbf {\bibinfo {volume} {49}},\ \bibinfo {pages} {1372–1393} (\bibinfo {year} {2019})}\BibitemShut {NoStop}%
\bibitem [{\citenamefont {Kaufman}\ \emph {et~al.}(2016)\citenamefont {Kaufman}, \citenamefont {Tai}, \citenamefont {Lukin}, \citenamefont {Rispoli}, \citenamefont {Schittko}, \citenamefont {Preiss},\ and\ \citenamefont {Greiner}}]{kaufman_2016_quantum}%
  \BibitemOpen
  \bibfield  {author} {\bibinfo {author} {\bibfnamefont {A.}~\bibnamefont {Kaufman}}, \bibinfo {author} {\bibfnamefont {M.~E.}\ \bibnamefont {Tai}}, \bibinfo {author} {\bibfnamefont {A.}~\bibnamefont {Lukin}}, \bibinfo {author} {\bibfnamefont {M.}~\bibnamefont {Rispoli}}, \bibinfo {author} {\bibfnamefont {R.}~\bibnamefont {Schittko}}, \bibinfo {author} {\bibfnamefont {P.~M.}\ \bibnamefont {Preiss}},\ and\ \bibinfo {author} {\bibfnamefont {M.}~\bibnamefont {Greiner}},\ }\bibfield  {title} {\bibinfo {title} {Quantum thermalization through entanglement in an isolated many-body system},\ }\href {https://doi.org/10.1126/science.aaf6725} {\bibfield  {journal} {\bibinfo  {journal} {Science}\ }\textbf {\bibinfo {volume} {353}},\ \bibinfo {pages} {794} (\bibinfo {year} {2016})}\BibitemShut {NoStop}%
\bibitem [{\citenamefont {Fleckenstein}\ and\ \citenamefont {Bukov}(2021)}]{bukov_many_body_2021}%
  \BibitemOpen
  \bibfield  {author} {\bibinfo {author} {\bibfnamefont {C.}~\bibnamefont {Fleckenstein}}\ and\ \bibinfo {author} {\bibfnamefont {M.}~\bibnamefont {Bukov}},\ }\bibfield  {title} {\bibinfo {title} {Prethermalization and thermalization in periodically driven many-body systems away from the high-frequency limit},\ }\href {https://doi.org/10.1103/PhysRevB.103.L140302} {\bibfield  {journal} {\bibinfo  {journal} {Physical Review B}\ }\textbf {\bibinfo {volume} {103}},\ \bibinfo {pages} {L140302} (\bibinfo {year} {2021})}\BibitemShut {NoStop}%
\bibitem [{\citenamefont {Kumari}\ and\ \citenamefont {Ghose}(2018)}]{Kumari_Ghose_vicinity_2018}%
  \BibitemOpen
  \bibfield  {author} {\bibinfo {author} {\bibfnamefont {M.}~\bibnamefont {Kumari}}\ and\ \bibinfo {author} {\bibfnamefont {S.}~\bibnamefont {Ghose}},\ }\bibfield  {title} {\bibinfo {title} {Quantum-classical correspondence in the vicinity of periodic orbits},\ }\href {https://doi.org/10.1103/PhysRevE.97.052209} {\bibfield  {journal} {\bibinfo  {journal} {Physical Review E}\ }\textbf {\bibinfo {volume} {97}},\ \bibinfo {pages} {052209} (\bibinfo {year} {2018})}\BibitemShut {NoStop}%
\bibitem [{\citenamefont {Pai}\ and\ \citenamefont {Pretko}(2019)}]{Pai_and_Pretko_scar_2019}%
  \BibitemOpen
  \bibfield  {author} {\bibinfo {author} {\bibfnamefont {S.}~\bibnamefont {Pai}}\ and\ \bibinfo {author} {\bibfnamefont {M.}~\bibnamefont {Pretko}},\ }\bibfield  {title} {\bibinfo {title} {Dynamical scar states in driven fracton systems},\ }\href {https://doi.org/10.1103/PhysRevLett.123.136401} {\bibfield  {journal} {\bibinfo  {journal} {Physical Review Letter}\ }\textbf {\bibinfo {volume} {123}},\ \bibinfo {pages} {136401} (\bibinfo {year} {2019})}\BibitemShut {NoStop}%
\bibitem [{\citenamefont {Mukherjee}\ \emph {et~al.}(2020)\citenamefont {Mukherjee}, \citenamefont {Nandy}, \citenamefont {Sen}, \citenamefont {Sen},\ and\ \citenamefont {Sengupta}}]{Sen_scar_2020}%
  \BibitemOpen
  \bibfield  {author} {\bibinfo {author} {\bibfnamefont {B.}~\bibnamefont {Mukherjee}}, \bibinfo {author} {\bibfnamefont {S.}~\bibnamefont {Nandy}}, \bibinfo {author} {\bibfnamefont {A.}~\bibnamefont {Sen}}, \bibinfo {author} {\bibfnamefont {D.}~\bibnamefont {Sen}},\ and\ \bibinfo {author} {\bibfnamefont {K.}~\bibnamefont {Sengupta}},\ }\bibfield  {title} {\bibinfo {title} {Collapse and revival of quantum many-body scars via floquet engineering},\ }\href {https://doi.org/10.1103/PhysRevB.101.245107} {\bibfield  {journal} {\bibinfo  {journal} {Physical Review B}\ }\textbf {\bibinfo {volume} {101}},\ \bibinfo {pages} {245107} (\bibinfo {year} {2020})}\BibitemShut {NoStop}%
\bibitem [{\citenamefont {Mizuta}\ \emph {et~al.}(2020)\citenamefont {Mizuta}, \citenamefont {Takasan},\ and\ \citenamefont {Kawakami}}]{Norio_scars_2020}%
  \BibitemOpen
  \bibfield  {author} {\bibinfo {author} {\bibfnamefont {K.}~\bibnamefont {Mizuta}}, \bibinfo {author} {\bibfnamefont {K.}~\bibnamefont {Takasan}},\ and\ \bibinfo {author} {\bibfnamefont {N.}~\bibnamefont {Kawakami}},\ }\bibfield  {title} {\bibinfo {title} {Exact floquet quantum many-body scars under rydberg blockade},\ }\href {https://doi.org/10.1103/PhysRevResearch.2.033284} {\bibfield  {journal} {\bibinfo  {journal} {Physical Review Research}\ }\textbf {\bibinfo {volume} {2}},\ \bibinfo {pages} {033284} (\bibinfo {year} {2020})}\BibitemShut {NoStop}%
\bibitem [{\citenamefont {Roy}\ and\ \citenamefont {Harper}(2017)}]{roy_2017_floquet}%
  \BibitemOpen
  \bibfield  {author} {\bibinfo {author} {\bibfnamefont {R.}~\bibnamefont {Roy}}\ and\ \bibinfo {author} {\bibfnamefont {F.}~\bibnamefont {Harper}},\ }\bibfield  {title} {\bibinfo {title} {Floquet topological phases with symmetry in all dimensions},\ }\href {https://journals.aps.org/prb/abstract/10.1103/PhysRevB.95.19512810.1103/physrevb.95.195128} {\bibfield  {journal} {\bibinfo  {journal} {Physical Review B}\ }\textbf {\bibinfo {volume} {95}} (\bibinfo {year} {2017})}\BibitemShut {NoStop}%
\bibitem [{\citenamefont {Sacha}\ and\ \citenamefont {Zakrzewski}(2017)}]{sacha_2017_time}%
  \BibitemOpen
  \bibfield  {author} {\bibinfo {author} {\bibfnamefont {K.}~\bibnamefont {Sacha}}\ and\ \bibinfo {author} {\bibfnamefont {J.}~\bibnamefont {Zakrzewski}},\ }\bibfield  {title} {\bibinfo {title} {Time crystals: a review},\ }\href {https://doi.org/10.1088/1361-6633/aa8b38} {\bibfield  {journal} {\bibinfo  {journal} {Reports on Progress in Physics}\ }\textbf {\bibinfo {volume} {81}},\ \bibinfo {pages} {016401} (\bibinfo {year} {2017})}\BibitemShut {NoStop}%
\bibitem [{\citenamefont {Khemani}\ \emph {et~al.}(2019)\citenamefont {Khemani}, \citenamefont {Moessner},\ and\ \citenamefont {Sondhi}}]{khemani_2019_a}%
  \BibitemOpen
  \bibfield  {author} {\bibinfo {author} {\bibfnamefont {V.}~\bibnamefont {Khemani}}, \bibinfo {author} {\bibfnamefont {R.}~\bibnamefont {Moessner}},\ and\ \bibinfo {author} {\bibfnamefont {S.~L.}\ \bibnamefont {Sondhi}},\ }\bibfield  {title} {\bibinfo {title} {A brief history of time crystals},\ }\href {https://arxiv.org/abs/1910.10745} {\bibfield  {journal} {\bibinfo  {journal} {arXiv}\ } (\bibinfo {year} {2019})}\BibitemShut {NoStop}%
\bibitem [{\citenamefont {Floquet}(1883)}]{gfloquet_1883}%
  \BibitemOpen
  \bibfield  {author} {\bibinfo {author} {\bibfnamefont {G.}~\bibnamefont {Floquet}},\ }\bibfield  {title} {\bibinfo {title} {Sur les équations différentielles linéaires à coefficients périodiques},\ }\href {https://doi.org/10.24033/asens.220} {\bibfield  {journal} {\bibinfo  {journal} {Annales Scientifiques De L Ecole Normale Superieure}\ }\textbf {\bibinfo {volume} {12}},\ \bibinfo {pages} {47} (\bibinfo {year} {1883})}\BibitemShut {NoStop}%
\bibitem [{\citenamefont {Grifoni}\ and\ \citenamefont {Hänggi}(1998)}]{grifoni_driven_1998}%
  \BibitemOpen
  \bibfield  {author} {\bibinfo {author} {\bibfnamefont {M.}~\bibnamefont {Grifoni}}\ and\ \bibinfo {author} {\bibfnamefont {P.}~\bibnamefont {Hänggi}},\ }\bibfield  {title} {\bibinfo {title} {Driven quantum tunneling},\ }\href {https://doi.org/10.1016/s0370-1573(98)00022-2} {\bibfield  {journal} {\bibinfo  {journal} {Physics Reports}\ }\textbf {\bibinfo {volume} {304}},\ \bibinfo {pages} {229} (\bibinfo {year} {1998})}\BibitemShut {NoStop}%
\bibitem [{\citenamefont {Chaudhury}\ \emph {et~al.}(2009)\citenamefont {Chaudhury}, \citenamefont {Smith}, \citenamefont {Anderson}, \citenamefont {Ghose},\ and\ \citenamefont {Jessen}}]{Chaudhury_2009}%
  \BibitemOpen
  \bibfield  {author} {\bibinfo {author} {\bibfnamefont {S.}~\bibnamefont {Chaudhury}}, \bibinfo {author} {\bibfnamefont {A.}~\bibnamefont {Smith}}, \bibinfo {author} {\bibfnamefont {B.~E.}\ \bibnamefont {Anderson}}, \bibinfo {author} {\bibfnamefont {S.}~\bibnamefont {Ghose}},\ and\ \bibinfo {author} {\bibfnamefont {P.~S.}\ \bibnamefont {Jessen}},\ }\bibfield  {title} {\bibinfo {title} {Quantum signatures of chaos in a kicked top},\ }\href {https://doi.org/10.1038/nature08396} {\bibfield  {journal} {\bibinfo  {journal} {Nature}\ }\textbf {\bibinfo {volume} {461}},\ \bibinfo {pages} {768–771} (\bibinfo {year} {2009})}\BibitemShut {NoStop}%
\bibitem [{\citenamefont {Neill}\ \emph {et~al.}(2016)\citenamefont {Neill}, \citenamefont {Roushan}, \citenamefont {Fang}, \citenamefont {Chen}, \citenamefont {Kolodrubetz}, \citenamefont {Chen}, \citenamefont {Megrant}, \citenamefont {Barends}, \citenamefont {Campbell}, \citenamefont {Chiaro} \emph {et~al.}}]{Neill_2016}%
  \BibitemOpen
  \bibfield  {author} {\bibinfo {author} {\bibfnamefont {C.}~\bibnamefont {Neill}}, \bibinfo {author} {\bibfnamefont {P.}~\bibnamefont {Roushan}}, \bibinfo {author} {\bibfnamefont {M.}~\bibnamefont {Fang}}, \bibinfo {author} {\bibfnamefont {Y.}~\bibnamefont {Chen}}, \bibinfo {author} {\bibfnamefont {M.}~\bibnamefont {Kolodrubetz}}, \bibinfo {author} {\bibfnamefont {Z.}~\bibnamefont {Chen}}, \bibinfo {author} {\bibfnamefont {A.}~\bibnamefont {Megrant}}, \bibinfo {author} {\bibfnamefont {R.}~\bibnamefont {Barends}}, \bibinfo {author} {\bibfnamefont {B.}~\bibnamefont {Campbell}}, \bibinfo {author} {\bibfnamefont {B.}~\bibnamefont {Chiaro}}, \emph {et~al.},\ }\bibfield  {title} {\bibinfo {title} {Ergodic dynamics and thermalization in an isolated quantum system},\ }\href@noop {} {\bibfield  {journal} {\bibinfo  {journal} {Nature Physics}\ }\textbf {\bibinfo {volume} {12}},\ \bibinfo {pages} {1037} (\bibinfo {year} {2016})}\BibitemShut {NoStop}%
\bibitem [{\citenamefont {Kanem}\ \emph {et~al.}(2007)\citenamefont {Kanem}, \citenamefont {Maneshi}, \citenamefont {Partlow}, \citenamefont {Spanner},\ and\ \citenamefont {Steinberg}}]{kanem_2007_observation}%
  \BibitemOpen
  \bibfield  {author} {\bibinfo {author} {\bibfnamefont {J.~F.}\ \bibnamefont {Kanem}}, \bibinfo {author} {\bibfnamefont {S.}~\bibnamefont {Maneshi}}, \bibinfo {author} {\bibfnamefont {M.}~\bibnamefont {Partlow}}, \bibinfo {author} {\bibfnamefont {M.}~\bibnamefont {Spanner}},\ and\ \bibinfo {author} {\bibfnamefont {A.~M.}\ \bibnamefont {Steinberg}},\ }\bibfield  {title} {\bibinfo {title} {Observation of high-order quantum resonances in the kicked rotor},\ }\bibfield  {journal} {\bibinfo  {journal} {Physical Review Letters}\ }\textbf {\bibinfo {volume} {98}},\ \href {https://doi.org/10.1103/physrevlett.98.083004} {10.1103/physrevlett.98.083004} (\bibinfo {year} {2007})\BibitemShut {NoStop}%
\bibitem [{\citenamefont {Krithika}\ \emph {et~al.}(2019)\citenamefont {Krithika}, \citenamefont {Anjusha}, \citenamefont {Bhosale},\ and\ \citenamefont {Mahesh}}]{Krithika_NMR_QKT_2019}%
  \BibitemOpen
  \bibfield  {author} {\bibinfo {author} {\bibfnamefont {V.~R.}\ \bibnamefont {Krithika}}, \bibinfo {author} {\bibfnamefont {V.~S.}\ \bibnamefont {Anjusha}}, \bibinfo {author} {\bibfnamefont {U.~T.}\ \bibnamefont {Bhosale}},\ and\ \bibinfo {author} {\bibfnamefont {T.~S.}\ \bibnamefont {Mahesh}},\ }\bibfield  {title} {\bibinfo {title} {{NMR} studies of quantum chaos in a two-qubit kicked top},\ }\href {https://doi.org/10.1103/PhysRevE.99.032219} {\bibfield  {journal} {\bibinfo  {journal} {Phys. Rev. E}\ }\textbf {\bibinfo {volume} {99}},\ \bibinfo {pages} {032219} (\bibinfo {year} {2019})}\BibitemShut {NoStop}%
\bibitem [{\citenamefont {Lignier}\ \emph {et~al.}(2007)\citenamefont {Lignier}, \citenamefont {Sias}, \citenamefont {Ciampini}, \citenamefont {Singh}, \citenamefont {Zenesini}, \citenamefont {Morsch},\ and\ \citenamefont {Arimondo}}]{Lignier2007}%
  \BibitemOpen
  \bibfield  {author} {\bibinfo {author} {\bibfnamefont {H.}~\bibnamefont {Lignier}}, \bibinfo {author} {\bibfnamefont {C.}~\bibnamefont {Sias}}, \bibinfo {author} {\bibfnamefont {D.}~\bibnamefont {Ciampini}}, \bibinfo {author} {\bibfnamefont {Y.}~\bibnamefont {Singh}}, \bibinfo {author} {\bibfnamefont {A.}~\bibnamefont {Zenesini}}, \bibinfo {author} {\bibfnamefont {O.}~\bibnamefont {Morsch}},\ and\ \bibinfo {author} {\bibfnamefont {E.}~\bibnamefont {Arimondo}},\ }\bibfield  {title} {\bibinfo {title} {Dynamical control of matter-wave tunneling in periodic potentials},\ }\href {https://doi.org/10.1103/PhysRevLett.99.220403} {\bibfield  {journal} {\bibinfo  {journal} {Physical Review Letters}\ }\textbf {\bibinfo {volume} {99}},\ \bibinfo {pages} {220403} (\bibinfo {year} {2007})}\BibitemShut {NoStop}%
\bibitem [{\citenamefont {Jotzu}\ \emph {et~al.}(2014)\citenamefont {Jotzu} \emph {et~al.}}]{Jotzu2014}%
  \BibitemOpen
  \bibfield  {author} {\bibinfo {author} {\bibfnamefont {G.}~\bibnamefont {Jotzu}} \emph {et~al.},\ }\bibfield  {title} {\bibinfo {title} {Experimental realization of the topological {Haldane} model with ultracold fermions},\ }\href {https://doi.org/10.1038/nature13915} {\bibfield  {journal} {\bibinfo  {journal} {Nature}\ }\textbf {\bibinfo {volume} {515}},\ \bibinfo {pages} {237} (\bibinfo {year} {2014})}\BibitemShut {NoStop}%
\bibitem [{\citenamefont {Mi}\ \emph {et~al.}(2022)\citenamefont {Mi} \emph {et~al.}}]{Mi2022TimeCrystal}%
  \BibitemOpen
  \bibfield  {author} {\bibinfo {author} {\bibfnamefont {X.}~\bibnamefont {Mi}} \emph {et~al.},\ }\bibfield  {title} {\bibinfo {title} {Time-crystalline eigenstate order on a quantum processor},\ }\href {https://doi.org/10.1038/s41586-021-04257-w} {\bibfield  {journal} {\bibinfo  {journal} {Nature}\ }\textbf {\bibinfo {volume} {601}},\ \bibinfo {pages} {531} (\bibinfo {year} {2022})}\BibitemShut {NoStop}%
\bibitem [{\citenamefont {Zhang}\ \emph {et~al.}(2022)\citenamefont {Zhang} \emph {et~al.}}]{Zhang2022FloquetSPT}%
  \BibitemOpen
  \bibfield  {author} {\bibinfo {author} {\bibfnamefont {X.}~\bibnamefont {Zhang}} \emph {et~al.},\ }\bibfield  {title} {\bibinfo {title} {Digital quantum simulation of {Floquet} symmetry-protected topological phases},\ }\href {https://doi.org/10.1038/s41586-022-04854-3} {\bibfield  {journal} {\bibinfo  {journal} {Nature}\ }\textbf {\bibinfo {volume} {607}},\ \bibinfo {pages} {468} (\bibinfo {year} {2022})}\BibitemShut {NoStop}%
\bibitem [{\citenamefont {Kim}\ \emph {et~al.}(2023{\natexlab{a}})\citenamefont {Kim} \emph {et~al.}}]{Kim2023Ising}%
  \BibitemOpen
  \bibfield  {author} {\bibinfo {author} {\bibfnamefont {Y.}~\bibnamefont {Kim}} \emph {et~al.},\ }\bibfield  {title} {\bibinfo {title} {Evidence for the utility of quantum computing before fault tolerance},\ }\href {https://doi.org/10.1038/s41586-023-06096-3} {\bibfield  {journal} {\bibinfo  {journal} {Nature}\ }\textbf {\bibinfo {volume} {618}},\ \bibinfo {pages} {500} (\bibinfo {year} {2023}{\natexlab{a}})}\BibitemShut {NoStop}%
\bibitem [{\citenamefont {Liao}\ \emph {et~al.}(2023)\citenamefont {Liao}, \citenamefont {Wang}, \citenamefont {Zhou}, \citenamefont {Zhang},\ and\ \citenamefont {Xiang}}]{liao2023simulation}%
  \BibitemOpen
  \bibfield  {author} {\bibinfo {author} {\bibfnamefont {H.-J.}\ \bibnamefont {Liao}}, \bibinfo {author} {\bibfnamefont {K.}~\bibnamefont {Wang}}, \bibinfo {author} {\bibfnamefont {Z.-S.}\ \bibnamefont {Zhou}}, \bibinfo {author} {\bibfnamefont {P.}~\bibnamefont {Zhang}},\ and\ \bibinfo {author} {\bibfnamefont {T.}~\bibnamefont {Xiang}},\ }\href {https://doi.org/10.48550/arXiv.2308.03082} {\bibinfo {title} {Simulation of {IBM}'s kicked {Ising} experiment with projected entangled pair operator}} (\bibinfo {year} {2023}),\ \Eprint {https://arxiv.org/abs/2308.03082} {arXiv:2308.03082 [quant-ph]} \BibitemShut {NoStop}%
\bibitem [{\citenamefont {Tindall}\ \emph {et~al.}(2024)\citenamefont {Tindall}, \citenamefont {Fishman}, \citenamefont {Stoudenmire},\ and\ \citenamefont {Sels}}]{tindall2024efficient}%
  \BibitemOpen
  \bibfield  {author} {\bibinfo {author} {\bibfnamefont {J.}~\bibnamefont {Tindall}}, \bibinfo {author} {\bibfnamefont {M.}~\bibnamefont {Fishman}}, \bibinfo {author} {\bibfnamefont {M.}~\bibnamefont {Stoudenmire}},\ and\ \bibinfo {author} {\bibfnamefont {D.}~\bibnamefont {Sels}},\ }\bibfield  {title} {\bibinfo {title} {Efficient tensor network simulation of {IBM}'s eagle kicked {Ising} experiment},\ }\href {https://doi.org/10.1103/PRXQuantum.5.010308} {\bibfield  {journal} {\bibinfo  {journal} {PRX Quantum}\ }\textbf {\bibinfo {volume} {5}},\ \bibinfo {pages} {010308} (\bibinfo {year} {2024})}\BibitemShut {NoStop}%
\bibitem [{\citenamefont {Hogg}\ and\ \citenamefont {Huberman}(1982)}]{hogg_1982_recurrence}%
  \BibitemOpen
  \bibfield  {author} {\bibinfo {author} {\bibfnamefont {T.}~\bibnamefont {Hogg}}\ and\ \bibinfo {author} {\bibfnamefont {B.~A.}\ \bibnamefont {Huberman}},\ }\bibfield  {title} {\bibinfo {title} {Recurrence phenomena in quantum dynamics},\ }\href {https://doi.org/10.1103/physrevlett.48.711} {\bibfield  {journal} {\bibinfo  {journal} {Physical Review Letters}\ }\textbf {\bibinfo {volume} {48}},\ \bibinfo {pages} {711} (\bibinfo {year} {1982})}\BibitemShut {NoStop}%
\bibitem [{\citenamefont {Pandit}\ \emph {et~al.}(2022)\citenamefont {Pandit}, \citenamefont {Green}, \citenamefont {Alderete}, \citenamefont {Linke},\ and\ \citenamefont {Uzdin}}]{pandit_2022_bounds}%
  \BibitemOpen
  \bibfield  {author} {\bibinfo {author} {\bibfnamefont {T.}~\bibnamefont {Pandit}}, \bibinfo {author} {\bibfnamefont {A.~M.}\ \bibnamefont {Green}}, \bibinfo {author} {\bibfnamefont {C.~H.}\ \bibnamefont {Alderete}}, \bibinfo {author} {\bibfnamefont {N.~M.}\ \bibnamefont {Linke}},\ and\ \bibinfo {author} {\bibfnamefont {R.}~\bibnamefont {Uzdin}},\ }\bibfield  {title} {\bibinfo {title} {Bounds on the recurrence probability in periodically-driven quantum systems},\ }\href {https://doi.org/10.22331/q-2022-04-06-682} {\bibfield  {journal} {\bibinfo  {journal} {Quantum}\ }\textbf {\bibinfo {volume} {6}},\ \bibinfo {pages} {682} (\bibinfo {year} {2022})}\BibitemShut {NoStop}%
\bibitem [{\citenamefont {Oszmaniec}\ \emph {et~al.}(2024)\citenamefont {Oszmaniec}, \citenamefont {Kotowski}, \citenamefont {Horodecki},\ and\ \citenamefont {Hunter-Jones}}]{oszmaniec_2024_saturation}%
  \BibitemOpen
  \bibfield  {author} {\bibinfo {author} {\bibfnamefont {M.}~\bibnamefont {Oszmaniec}}, \bibinfo {author} {\bibfnamefont {M.}~\bibnamefont {Kotowski}}, \bibinfo {author} {\bibfnamefont {M.}~\bibnamefont {Horodecki}},\ and\ \bibinfo {author} {\bibfnamefont {N.}~\bibnamefont {Hunter-Jones}},\ }\bibfield  {title} {\bibinfo {title} {Saturation and recurrence of quantum complexity in random local quantum dynamics},\ }\href {https://journals.aps.org/prx/abstract/10.1103/PhysRevX.14.04106810.1103/physrevx.14.041068} {\bibfield  {journal} {\bibinfo  {journal} {Physical Review X}\ }\textbf {\bibinfo {volume} {14}} (\bibinfo {year} {2024})}\BibitemShut {NoStop}%
\bibitem [{\citenamefont {Poincaré}(1890)}]{poincar_1890_avantpropos}%
  \BibitemOpen
  \bibfield  {author} {\bibinfo {author} {\bibfnamefont {H.}~\bibnamefont {Poincaré}},\ }\bibfield  {title} {\bibinfo {title} {Avant-propos},\ }\href {https://doi.org/10.1007/bf02392505} {\bibfield  {journal} {\bibinfo  {journal} {Acta Mathematica}\ }\textbf {\bibinfo {volume} {13}},\ \bibinfo {pages} {VII} (\bibinfo {year} {1890})}\BibitemShut {NoStop}%
\bibitem [{\citenamefont {Bocchieri}\ and\ \citenamefont {Loinger}(1957)}]{bocchieri_1957_quantum}%
  \BibitemOpen
  \bibfield  {author} {\bibinfo {author} {\bibfnamefont {P.}~\bibnamefont {Bocchieri}}\ and\ \bibinfo {author} {\bibfnamefont {A.}~\bibnamefont {Loinger}},\ }\bibfield  {title} {\bibinfo {title} {Quantum recurrence theorem},\ }\href {https://journals.aps.org/pr/abstract/10.1103/PhysRev.107.337} {\bibfield  {journal} {\bibinfo  {journal} {Physical Review}\ }\textbf {\bibinfo {volume} {107}},\ \bibinfo {pages} {337} (\bibinfo {year} {1957})}\BibitemShut {NoStop}%
\bibitem [{\citenamefont {Gupta}\ and\ \citenamefont {Short}(2026)}]{Chaitanya2026recurrence}%
  \BibitemOpen
  \bibfield  {author} {\bibinfo {author} {\bibfnamefont {C.}~\bibnamefont {Gupta}}\ and\ \bibinfo {author} {\bibfnamefont {A.~J.}\ \bibnamefont {Short}},\ }\bibfield  {title} {\bibinfo {title} {Recurrence time for finite quantum systems},\ }\href {https://doi.org/10.1103/392k-bgpt} {\bibfield  {journal} {\bibinfo  {journal} {Phys. Rev. Lett.}\ }\textbf {\bibinfo {volume} {137}},\ \bibinfo {pages} {140201} (\bibinfo {year} {2026})}\BibitemShut {NoStop}%
\bibitem [{\citenamefont {Gimeno}\ and\ \citenamefont {Sotoca}(2017)}]{gimeno_upper_2017}%
  \BibitemOpen
  \bibfield  {author} {\bibinfo {author} {\bibfnamefont {V.}~\bibnamefont {Gimeno}}\ and\ \bibinfo {author} {\bibfnamefont {J.~M.}\ \bibnamefont {Sotoca}},\ }\bibfield  {title} {\bibinfo {title} {Upper bounds for the poincaré recurrence time in quantum mixed states},\ }\href {https://iopscience.iop.org/article/10.1088/1751-8121/aa67fe} {\bibfield  {journal} {\bibinfo  {journal} {Journal of Physics A Mathematical and Theoretical}\ }\textbf {\bibinfo {volume} {50}} (\bibinfo {year} {2017})}\BibitemShut {NoStop}%
\bibitem [{\citenamefont {Brown}\ and\ \citenamefont {Susskind}(2018)}]{suskind_complexity_2018}%
  \BibitemOpen
  \bibfield  {author} {\bibinfo {author} {\bibfnamefont {A.~R.}\ \bibnamefont {Brown}}\ and\ \bibinfo {author} {\bibfnamefont {L.}~\bibnamefont {Susskind}},\ }\bibfield  {title} {\bibinfo {title} {Second law of quantum complexity},\ }\href {https://doi.org/10.1103/PhysRevD.97.086015} {\bibfield  {journal} {\bibinfo  {journal} {Physical Review D}\ }\textbf {\bibinfo {volume} {97}},\ \bibinfo {pages} {086015} (\bibinfo {year} {2018})}\BibitemShut {NoStop}%
\bibitem [{\citenamefont {Freedman}(2024)}]{freedman_quantumrecurrent_2024}%
  \BibitemOpen
  \bibfield  {author} {\bibinfo {author} {\bibfnamefont {M.~H.}\ \bibnamefont {Freedman}},\ }\bibfield  {title} {\bibinfo {title} {Quantum detection of recurrent dynamics},\ }\href {https://doi.org/10.48550/arXiv.2407.16055} {\bibfield  {journal} {\bibinfo  {journal} {arXiv:2407.16055}\ } (\bibinfo {year} {2024})},\ \Eprint {https://arxiv.org/abs/2407.16055} {arXiv 2407.16055:2407.16055 [quant-ph]} \BibitemShut {NoStop}%
\bibitem [{\citenamefont {Levesque}\ and\ \citenamefont {Sourlas}(2025)}]{levesque_rec_time_2025}%
  \BibitemOpen
  \bibfield  {author} {\bibinfo {author} {\bibfnamefont {D.}~\bibnamefont {Levesque}}\ and\ \bibinfo {author} {\bibfnamefont {N.}~\bibnamefont {Sourlas}},\ }\bibfield  {title} {\bibinfo {title} {Time irreversibility in statistical mechanics},\ }\href {https://hal.science/hal-0515352510.1007/s10955-025-03467-0} {\bibfield  {journal} {\bibinfo  {journal} {Journal of Statistical Physics}\ }\textbf {\bibinfo {volume} {192}} (\bibinfo {year} {2025})}\BibitemShut {NoStop}%
\bibitem [{\citenamefont {Ropotenko}(2025)}]{ropotenko_the_2025}%
  \BibitemOpen
  \bibfield  {author} {\bibinfo {author} {\bibfnamefont {K.}~\bibnamefont {Ropotenko}},\ }\bibfield  {title} {\bibinfo {title} {The poincar\'e recurrence time for the de sitter space with dynamical chaos},\ }\href {https://arxiv.org/abs/0712.0993} {\bibfield  {journal} {\bibinfo  {journal} {arXiv}\ } (\bibinfo {year} {2025})}\BibitemShut {NoStop}%
\bibitem [{\citenamefont {Sharma}\ and\ \citenamefont {Bhosale}(2024)}]{bhosle_qkt__jan_2024}%
  \BibitemOpen
  \bibfield  {author} {\bibinfo {author} {\bibfnamefont {H.}~\bibnamefont {Sharma}}\ and\ \bibinfo {author} {\bibfnamefont {U.~T.}\ \bibnamefont {Bhosale}},\ }\bibfield  {title} {\bibinfo {title} {Exactly solvable dynamics and signatures of integrability in an infinite-range many-body floquet spin system},\ }\href {https://doi.org/10.1103/PhysRevB.109.014412} {\bibfield  {journal} {\bibinfo  {journal} {Physical Review B}\ }\textbf {\bibinfo {volume} {109}},\ \bibinfo {pages} {014412} (\bibinfo {year} {2024})}\BibitemShut {NoStop}%
\bibitem [{\citenamefont {Purohit}\ and\ \citenamefont {Bhosale}(2026)}]{Purohit2026DoubleKickedTop}%
  \BibitemOpen
  \bibfield  {author} {\bibinfo {author} {\bibfnamefont {A.~V.}\ \bibnamefont {Purohit}}\ and\ \bibinfo {author} {\bibfnamefont {U.~T.}\ \bibnamefont {Bhosale}},\ }\bibfield  {title} {\bibinfo {title} {{Floquet} resonances in the double kicked top},\ }\href {https://doi.org/10.1103/nk7g-ckb8} {\bibfield  {journal} {\bibinfo  {journal} {Physical Review E}\ }\textbf {\bibinfo {volume} {114}},\ \bibinfo {pages} {014214} (\bibinfo {year} {2026})}\BibitemShut {NoStop}%
\bibitem [{\citenamefont {von Keyserlingk}\ \emph {et~al.}(2018)\citenamefont {von Keyserlingk}, \citenamefont {Rakovszky}, \citenamefont {Pollmann},\ and\ \citenamefont {Sondhi}}]{vonKeyserlingk2018Operator}%
  \BibitemOpen
  \bibfield  {author} {\bibinfo {author} {\bibfnamefont {C.~W.}\ \bibnamefont {von Keyserlingk}}, \bibinfo {author} {\bibfnamefont {T.}~\bibnamefont {Rakovszky}}, \bibinfo {author} {\bibfnamefont {F.}~\bibnamefont {Pollmann}},\ and\ \bibinfo {author} {\bibfnamefont {S.~L.}\ \bibnamefont {Sondhi}},\ }\bibfield  {title} {\bibinfo {title} {Operator hydrodynamics, {OTOCs}, and entanglement growth in systems without conservation laws},\ }\href {https://doi.org/10.1103/PhysRevX.8.021013} {\bibfield  {journal} {\bibinfo  {journal} {Physical Review X}\ }\textbf {\bibinfo {volume} {8}},\ \bibinfo {pages} {021013} (\bibinfo {year} {2018})},\ \Eprint {https://arxiv.org/abs/1705.08910} {arXiv:1705.08910 [cond-mat.str-el]} \BibitemShut {NoStop}%
\bibitem [{\citenamefont {Zou}\ \emph {et~al.}(2025)\citenamefont {Zou}, \citenamefont {Gong},\ and\ \citenamefont {Chen}}]{zou_2025_enhancing}%
  \BibitemOpen
  \bibfield  {author} {\bibinfo {author} {\bibfnamefont {Z.}~\bibnamefont {Zou}}, \bibinfo {author} {\bibfnamefont {J.}~\bibnamefont {Gong}},\ and\ \bibinfo {author} {\bibfnamefont {W.}~\bibnamefont {Chen}},\ }\bibfield  {title} {\bibinfo {title} {Enhancing quantum metrology by quantum resonance dynamics},\ }\href {https://journals.aps.org/prl/abstract/10.1103/lkrt-lvng} {\bibfield  {journal} {\bibinfo  {journal} {Physical Review Letters}\ }\textbf {\bibinfo {volume} {134}} (\bibinfo {year} {2025})}\BibitemShut {NoStop}%
\bibitem [{\citenamefont {Biswas}\ and\ \citenamefont {Choudhury}(2026)}]{biswas_2025_the}%
  \BibitemOpen
  \bibfield  {author} {\bibinfo {author} {\bibfnamefont {H.}~\bibnamefont {Biswas}}\ and\ \bibinfo {author} {\bibfnamefont {S.}~\bibnamefont {Choudhury}},\ }\bibfield  {title} {\bibinfo {title} {Floquet central spin model: A platform to realize time crystals, entanglement steering, and multiparameter metrology},\ }\href {https://doi.org/10.1103/vhx7-gbpt} {\bibfield  {journal} {\bibinfo  {journal} {Phys. Rev. B}\ }\textbf {\bibinfo {volume} {114}},\ \bibinfo {pages} {L140303} (\bibinfo {year} {2026})}\BibitemShut {NoStop}%
\bibitem [{\citenamefont {{Sreeram PG}}\ and\ \citenamefont {Santhanam}(2026)}]{SreeramSanthanam2026}%
  \BibitemOpen
  \bibfield  {author} {\bibinfo {author} {\bibnamefont {{Sreeram PG}}}\ and\ \bibinfo {author} {\bibfnamefont {M.~S.}\ \bibnamefont {Santhanam}},\ }\href {https://doi.org/10.48550/arXiv.2606.01953} {\bibinfo {title} {Quantum resonance based encryption protocol with quantum kicked top}} (\bibinfo {year} {2026}),\ \Eprint {https://arxiv.org/abs/2606.01953} {arXiv:2606.01953 [quant-ph]} \BibitemShut {NoStop}%
\bibitem [{\citenamefont {Romero}\ \emph {et~al.}(2026)\citenamefont {Romero}, \citenamefont {Chen},\ and\ \citenamefont {Ban}}]{Romero2026Battery}%
  \BibitemOpen
  \bibfield  {author} {\bibinfo {author} {\bibfnamefont {S.~V.}\ \bibnamefont {Romero}}, \bibinfo {author} {\bibfnamefont {X.}~\bibnamefont {Chen}},\ and\ \bibinfo {author} {\bibfnamefont {Y.}~\bibnamefont {Ban}},\ }\bibfield  {title} {\bibinfo {title} {Kicked-ising quantum battery},\ }\href {https://doi.org/10.1103/j35s-xv5k} {\bibfield  {journal} {\bibinfo  {journal} {Physical Review Letters}\ }\textbf {\bibinfo {volume} {137}},\ \bibinfo {pages} {050404} (\bibinfo {year} {2026})}\BibitemShut {NoStop}%
\bibitem [{\citenamefont {Haake}\ \emph {et~al.}(1987)\citenamefont {Haake}, \citenamefont {Kuś},\ and\ \citenamefont {Scharf}}]{haake_1987}%
  \BibitemOpen
  \bibfield  {author} {\bibinfo {author} {\bibfnamefont {F.}~\bibnamefont {Haake}}, \bibinfo {author} {\bibfnamefont {M.}~\bibnamefont {Kuś}},\ and\ \bibinfo {author} {\bibfnamefont {R.}~\bibnamefont {Scharf}},\ }\bibfield  {title} {\bibinfo {title} {Classical and quantum chaos for a kicked top},\ }\href {https://api.semanticscholar.org/CorpusID:123871040} {\bibfield  {journal} {\bibinfo  {journal} {Zeitschrift f{\"u}r Physik B Condensed Matter}\ }\textbf {\bibinfo {volume} {65}},\ \bibinfo {pages} {381} (\bibinfo {year} {1987})}\BibitemShut {NoStop}%
\bibitem [{\citenamefont {Ho}\ and\ \citenamefont {Choi}(2022)}]{ho2022exact}%
  \BibitemOpen
  \bibfield  {author} {\bibinfo {author} {\bibfnamefont {W.~W.}\ \bibnamefont {Ho}}\ and\ \bibinfo {author} {\bibfnamefont {S.}~\bibnamefont {Choi}},\ }\bibfield  {title} {\bibinfo {title} {Exact emergent quantum state designs from quantum chaotic dynamics},\ }\href {https://doi.org/10.1103/PhysRevLett.128.060601} {\bibfield  {journal} {\bibinfo  {journal} {Physical Review Letters}\ }\textbf {\bibinfo {volume} {128}},\ \bibinfo {pages} {060601} (\bibinfo {year} {2022})}\BibitemShut {NoStop}%
\bibitem [{\citenamefont {Prosen}(1998)}]{Prosen1998Integrable}%
  \BibitemOpen
  \bibfield  {author} {\bibinfo {author} {\bibfnamefont {T.}~\bibnamefont {Prosen}},\ }\bibfield  {title} {\bibinfo {title} {A new class of completely integrable quantum spin chains},\ }\href {https://doi.org/10.1088/0305-4470/31/21/002} {\bibfield  {journal} {\bibinfo  {journal} {Journal of Physics A: Mathematical and General}\ }\textbf {\bibinfo {volume} {31}},\ \bibinfo {pages} {L397} (\bibinfo {year} {1998})}\BibitemShut {NoStop}%
\bibitem [{\citenamefont {Prosen}(2000)}]{Prosen2000KickedIsing}%
  \BibitemOpen
  \bibfield  {author} {\bibinfo {author} {\bibfnamefont {T.}~\bibnamefont {Prosen}},\ }\bibfield  {title} {\bibinfo {title} {Exact time-correlation functions of quantum ising chain in a kicking transversal magnetic field: Spectral analysis of the adjoint propagator in heisenberg picture},\ }\href {https://doi.org/10.1143/PTPS.139.191} {\bibfield  {journal} {\bibinfo  {journal} {Progress of Theoretical Physics Supplement}\ }\textbf {\bibinfo {volume} {139}},\ \bibinfo {pages} {191} (\bibinfo {year} {2000})}\BibitemShut {NoStop}%
\bibitem [{\citenamefont {Bertini}\ \emph {et~al.}(2019)\citenamefont {Bertini}, \citenamefont {Kos},\ and\ \citenamefont {Prosen}}]{bertini2019entanglement}%
  \BibitemOpen
  \bibfield  {author} {\bibinfo {author} {\bibfnamefont {B.}~\bibnamefont {Bertini}}, \bibinfo {author} {\bibfnamefont {P.}~\bibnamefont {Kos}},\ and\ \bibinfo {author} {\bibfnamefont {T.}~\bibnamefont {Prosen}},\ }\bibfield  {title} {\bibinfo {title} {Entanglement spreading in a minimal model of maximal many-body quantum chaos},\ }\href {https://doi.org/10.1103/PhysRevX.9.021033} {\bibfield  {journal} {\bibinfo  {journal} {Physical Review X}\ }\textbf {\bibinfo {volume} {9}},\ \bibinfo {pages} {021033} (\bibinfo {year} {2019})}\BibitemShut {NoStop}%
\bibitem [{\citenamefont {Akila}\ \emph {et~al.}(2016)\citenamefont {Akila}, \citenamefont {Waltner}, \citenamefont {Gutkin},\ and\ \citenamefont {Guhr}}]{akila2016particle}%
  \BibitemOpen
  \bibfield  {author} {\bibinfo {author} {\bibfnamefont {M.}~\bibnamefont {Akila}}, \bibinfo {author} {\bibfnamefont {D.}~\bibnamefont {Waltner}}, \bibinfo {author} {\bibfnamefont {B.}~\bibnamefont {Gutkin}},\ and\ \bibinfo {author} {\bibfnamefont {T.}~\bibnamefont {Guhr}},\ }\bibfield  {title} {\bibinfo {title} {Particle-time duality in the kicked {Ising} spin chain},\ }\href {https://doi.org/10.1088/1751-8113/49/37/375101} {\bibfield  {journal} {\bibinfo  {journal} {Journal of Physics A: Mathematical and Theoretical}\ }\textbf {\bibinfo {volume} {49}},\ \bibinfo {pages} {375101} (\bibinfo {year} {2016})}\BibitemShut {NoStop}%
\bibitem [{\citenamefont {Bertini}\ \emph {et~al.}(2018)\citenamefont {Bertini}, \citenamefont {Kos},\ and\ \citenamefont {Prosen}}]{bertini2018exact}%
  \BibitemOpen
  \bibfield  {author} {\bibinfo {author} {\bibfnamefont {B.}~\bibnamefont {Bertini}}, \bibinfo {author} {\bibfnamefont {P.}~\bibnamefont {Kos}},\ and\ \bibinfo {author} {\bibfnamefont {T.}~\bibnamefont {Prosen}},\ }\bibfield  {title} {\bibinfo {title} {Exact spectral form factor in a minimal model of many-body quantum chaos},\ }\href {https://doi.org/10.1103/PhysRevLett.121.264101} {\bibfield  {journal} {\bibinfo  {journal} {Physical Review Letters}\ }\textbf {\bibinfo {volume} {121}},\ \bibinfo {pages} {264101} (\bibinfo {year} {2018})}\BibitemShut {NoStop}%
\bibitem [{\citenamefont {Gopalakrishnan}\ and\ \citenamefont {Lamacraft}(2019)}]{gopalakrishnan2019unitary}%
  \BibitemOpen
  \bibfield  {author} {\bibinfo {author} {\bibfnamefont {S.}~\bibnamefont {Gopalakrishnan}}\ and\ \bibinfo {author} {\bibfnamefont {A.}~\bibnamefont {Lamacraft}},\ }\bibfield  {title} {\bibinfo {title} {Unitary circuits of finite depth and infinite width from quantum channels},\ }\href {https://doi.org/10.1103/PhysRevB.100.064309} {\bibfield  {journal} {\bibinfo  {journal} {Physical Review B}\ }\textbf {\bibinfo {volume} {100}},\ \bibinfo {pages} {064309} (\bibinfo {year} {2019})}\BibitemShut {NoStop}%
\bibitem [{\citenamefont {Kim}\ \emph {et~al.}(2023{\natexlab{b}})\citenamefont {Kim}, \citenamefont {Eddins}, \citenamefont {Anand}, \citenamefont {Wei}, \citenamefont {van~den Berg}, \citenamefont {Rosenblatt}, \citenamefont {Nayfeh}, \citenamefont {Wu}, \citenamefont {Zaletel}, \citenamefont {Temme},\ and\ \citenamefont {Kandala}}]{kim2023evidence}%
  \BibitemOpen
  \bibfield  {author} {\bibinfo {author} {\bibfnamefont {Y.}~\bibnamefont {Kim}}, \bibinfo {author} {\bibfnamefont {A.}~\bibnamefont {Eddins}}, \bibinfo {author} {\bibfnamefont {S.}~\bibnamefont {Anand}}, \bibinfo {author} {\bibfnamefont {K.~X.}\ \bibnamefont {Wei}}, \bibinfo {author} {\bibfnamefont {E.}~\bibnamefont {van~den Berg}}, \bibinfo {author} {\bibfnamefont {S.}~\bibnamefont {Rosenblatt}}, \bibinfo {author} {\bibfnamefont {H.}~\bibnamefont {Nayfeh}}, \bibinfo {author} {\bibfnamefont {Y.}~\bibnamefont {Wu}}, \bibinfo {author} {\bibfnamefont {M.}~\bibnamefont {Zaletel}}, \bibinfo {author} {\bibfnamefont {K.}~\bibnamefont {Temme}},\ and\ \bibinfo {author} {\bibfnamefont {A.}~\bibnamefont {Kandala}},\ }\bibfield  {title} {\bibinfo {title} {Evidence for the utility of quantum computing before fault tolerance},\ }\href {https://doi.org/10.1038/s41586-023-06096-3} {\bibfield  {journal} {\bibinfo  {journal} {Nature}\ }\textbf {\bibinfo {volume} {618}},\ \bibinfo {pages} {500} (\bibinfo {year}
  {2023}{\natexlab{b}})}\BibitemShut {NoStop}%
\bibitem [{\citenamefont {Begu{\v{s}}i{\'c}}\ \emph {et~al.}(2024)\citenamefont {Begu{\v{s}}i{\'c}}, \citenamefont {Gray},\ and\ \citenamefont {Chan}}]{begusic2024fast}%
  \BibitemOpen
  \bibfield  {author} {\bibinfo {author} {\bibfnamefont {T.}~\bibnamefont {Begu{\v{s}}i{\'c}}}, \bibinfo {author} {\bibfnamefont {J.}~\bibnamefont {Gray}},\ and\ \bibinfo {author} {\bibfnamefont {G.~K.-L.}\ \bibnamefont {Chan}},\ }\bibfield  {title} {\bibinfo {title} {Fast and converged classical simulations of evidence for the utility of quantum computing before fault tolerance},\ }\href {https://doi.org/10.1126/sciadv.adk4321} {\bibfield  {journal} {\bibinfo  {journal} {Science Advances}\ }\textbf {\bibinfo {volume} {10}},\ \bibinfo {pages} {eadk4321} (\bibinfo {year} {2024})}\BibitemShut {NoStop}%
\bibitem [{\citenamefont {Tirrito}\ \emph {et~al.}(2025)\citenamefont {Tirrito}, \citenamefont {Turkeshi},\ and\ \citenamefont {Sierant}}]{tirrito2024anticoncentration}%
  \BibitemOpen
  \bibfield  {author} {\bibinfo {author} {\bibfnamefont {E.}~\bibnamefont {Tirrito}}, \bibinfo {author} {\bibfnamefont {X.}~\bibnamefont {Turkeshi}},\ and\ \bibinfo {author} {\bibfnamefont {P.}~\bibnamefont {Sierant}},\ }\bibfield  {title} {\bibinfo {title} {Anticoncentration and nonstabilizerness spreading under ergodic quantum dynamics},\ }\href {https://doi.org/10.1103/1jzy-sk9r} {\bibfield  {journal} {\bibinfo  {journal} {Physical Review Letters}\ }\textbf {\bibinfo {volume} {135}},\ \bibinfo {pages} {220401} (\bibinfo {year} {2025})},\ \Eprint {https://arxiv.org/abs/2412.10229} {arXiv:2412.10229 [quant-ph]} \BibitemShut {NoStop}%
\bibitem [{\citenamefont {Jha}\ and\ \citenamefont {Gadge}(2026)}]{jha2026magic}%
  \BibitemOpen
  \bibfield  {author} {\bibinfo {author} {\bibfnamefont {R.}~\bibnamefont {Jha}}\ and\ \bibinfo {author} {\bibfnamefont {K.}~\bibnamefont {Gadge}},\ }\href {https://doi.org/10.48550/arXiv.2607.14222} {\bibinfo {title} {Magic without entanglement: Exact revivals and their fisher information origin}} (\bibinfo {year} {2026}),\ \Eprint {https://arxiv.org/abs/2607.14222} {arXiv:2607.14222 [quant-ph]} \BibitemShut {NoStop}%
\bibitem [{\citenamefont {Axler}(2024)}]{Axler2024LinearAlgebra}%
  \BibitemOpen
  \bibfield  {author} {\bibinfo {author} {\bibfnamefont {S.}~\bibnamefont {Axler}},\ }\href {https://doi.org/10.1007/978-3-031-41026-0} {\emph {\bibinfo {title} {Linear Algebra Done Right}}},\ \bibinfo {edition} {4th}\ ed.,\ Undergraduate Texts in Mathematics\ (\bibinfo  {publisher} {Springer},\ \bibinfo {year} {2024})\BibitemShut {NoStop}%
\bibitem [{\citenamefont {Flack}\ \emph {et~al.}(2020)\citenamefont {Flack}, \citenamefont {Bertini},\ and\ \citenamefont {Prosen}}]{flack2020statistics}%
  \BibitemOpen
  \bibfield  {author} {\bibinfo {author} {\bibfnamefont {A.}~\bibnamefont {Flack}}, \bibinfo {author} {\bibfnamefont {B.}~\bibnamefont {Bertini}},\ and\ \bibinfo {author} {\bibfnamefont {T.}~\bibnamefont {Prosen}},\ }\bibfield  {title} {\bibinfo {title} {Statistics of the spectral form factor in the self-dual kicked {Ising} model},\ }\href {https://doi.org/10.1103/PhysRevResearch.2.043403} {\bibfield  {journal} {\bibinfo  {journal} {Physical Review Research}\ }\textbf {\bibinfo {volume} {2}},\ \bibinfo {pages} {043403} (\bibinfo {year} {2020})}\BibitemShut {NoStop}%
\bibitem [{\citenamefont {Braun}\ \emph {et~al.}(2020)\citenamefont {Braun}, \citenamefont {Waltner}, \citenamefont {Akila}, \citenamefont {Gutkin},\ and\ \citenamefont {Guhr}}]{braun2020transition}%
  \BibitemOpen
  \bibfield  {author} {\bibinfo {author} {\bibfnamefont {P.}~\bibnamefont {Braun}}, \bibinfo {author} {\bibfnamefont {D.}~\bibnamefont {Waltner}}, \bibinfo {author} {\bibfnamefont {M.}~\bibnamefont {Akila}}, \bibinfo {author} {\bibfnamefont {B.}~\bibnamefont {Gutkin}},\ and\ \bibinfo {author} {\bibfnamefont {T.}~\bibnamefont {Guhr}},\ }\bibfield  {title} {\bibinfo {title} {Transition from quantum chaos to localization in spin chains},\ }\href {https://doi.org/10.1103/PhysRevE.101.052201} {\bibfield  {journal} {\bibinfo  {journal} {Physical Review E}\ }\textbf {\bibinfo {volume} {101}},\ \bibinfo {pages} {052201} (\bibinfo {year} {2020})}\BibitemShut {NoStop}%
\bibitem [{\citenamefont {Sierant}\ \emph {et~al.}(2023)\citenamefont {Sierant}, \citenamefont {Lewenstein}, \citenamefont {Scardicchio},\ and\ \citenamefont {Zakrzewski}}]{sierant2023stability}%
  \BibitemOpen
  \bibfield  {author} {\bibinfo {author} {\bibfnamefont {P.}~\bibnamefont {Sierant}}, \bibinfo {author} {\bibfnamefont {M.}~\bibnamefont {Lewenstein}}, \bibinfo {author} {\bibfnamefont {A.}~\bibnamefont {Scardicchio}},\ and\ \bibinfo {author} {\bibfnamefont {J.}~\bibnamefont {Zakrzewski}},\ }\bibfield  {title} {\bibinfo {title} {Stability of many-body localization in floquet systems},\ }\href {https://doi.org/10.1103/PhysRevB.107.115132} {\bibfield  {journal} {\bibinfo  {journal} {Physical Review B}\ }\textbf {\bibinfo {volume} {107}},\ \bibinfo {pages} {115132} (\bibinfo {year} {2023})}\BibitemShut {NoStop}%
\bibitem [{\citenamefont {Herrmann}\ \emph {et~al.}(2023)\citenamefont {Herrmann}, \citenamefont {Kieler},\ and\ \citenamefont {B{\"a}cker}}]{herrmann2023characterizing}%
  \BibitemOpen
  \bibfield  {author} {\bibinfo {author} {\bibfnamefont {T.}~\bibnamefont {Herrmann}}, \bibinfo {author} {\bibfnamefont {M.~F.~I.}\ \bibnamefont {Kieler}},\ and\ \bibinfo {author} {\bibfnamefont {A.}~\bibnamefont {B{\"a}cker}},\ }\bibfield  {title} {\bibinfo {title} {Characterizing quantum chaoticity of kicked spin chains},\ }\href {https://doi.org/10.1103/PhysRevE.108.044213} {\bibfield  {journal} {\bibinfo  {journal} {Physical Review E}\ }\textbf {\bibinfo {volume} {108}},\ \bibinfo {pages} {044213} (\bibinfo {year} {2023})}\BibitemShut {NoStop}%
\bibitem [{\citenamefont {Mishra}\ \emph {et~al.}(2015)\citenamefont {Mishra}, \citenamefont {Lakshminarayan},\ and\ \citenamefont {Subrahmanyam}}]{mishra2015protocol}%
  \BibitemOpen
  \bibfield  {author} {\bibinfo {author} {\bibfnamefont {S.~K.}\ \bibnamefont {Mishra}}, \bibinfo {author} {\bibfnamefont {A.}~\bibnamefont {Lakshminarayan}},\ and\ \bibinfo {author} {\bibfnamefont {V.}~\bibnamefont {Subrahmanyam}},\ }\bibfield  {title} {\bibinfo {title} {Protocol using kicked {Ising} dynamics for generating states with maximal multipartite entanglement},\ }\href {https://doi.org/10.1103/PhysRevA.91.022318} {\bibfield  {journal} {\bibinfo  {journal} {Physical Review A}\ }\textbf {\bibinfo {volume} {91}},\ \bibinfo {pages} {022318} (\bibinfo {year} {2015})}\BibitemShut {NoStop}%
\bibitem [{\citenamefont {Gottesman}(1998)}]{Gottesman1998FaultTolerant}%
  \BibitemOpen
  \bibfield  {author} {\bibinfo {author} {\bibfnamefont {D.}~\bibnamefont {Gottesman}},\ }\bibfield  {title} {\bibinfo {title} {Theory of fault-tolerant quantum computation},\ }\href {https://doi.org/10.1103/PhysRevA.57.127} {\bibfield  {journal} {\bibinfo  {journal} {Physical Review A}\ }\textbf {\bibinfo {volume} {57}},\ \bibinfo {pages} {127} (\bibinfo {year} {1998})}\BibitemShut {NoStop}%
\bibitem [{\citenamefont {Leone}\ \emph {et~al.}(2022)\citenamefont {Leone}, \citenamefont {Oliviero},\ and\ \citenamefont {Hamma}}]{Leone2022Stabilizer}%
  \BibitemOpen
  \bibfield  {author} {\bibinfo {author} {\bibfnamefont {L.}~\bibnamefont {Leone}}, \bibinfo {author} {\bibfnamefont {S.~F.~E.}\ \bibnamefont {Oliviero}},\ and\ \bibinfo {author} {\bibfnamefont {A.}~\bibnamefont {Hamma}},\ }\bibfield  {title} {\bibinfo {title} {Stabilizer r{\'e}nyi entropy},\ }\href {https://doi.org/10.1103/PhysRevLett.128.050402} {\bibfield  {journal} {\bibinfo  {journal} {Physical Review Letters}\ }\textbf {\bibinfo {volume} {128}},\ \bibinfo {pages} {050402} (\bibinfo {year} {2022})}\BibitemShut {NoStop}%
\bibitem [{\citenamefont {Huang}\ \emph {et~al.}(2025)\citenamefont {Huang}, \citenamefont {Li}, \citenamefont {Lee},\ and\ \citenamefont {Zhong}}]{Huang2025XORFWHT}%
  \BibitemOpen
  \bibfield  {author} {\bibinfo {author} {\bibfnamefont {X.}~\bibnamefont {Huang}}, \bibinfo {author} {\bibfnamefont {H.-Z.}\ \bibnamefont {Li}}, \bibinfo {author} {\bibfnamefont {C.~H.}\ \bibnamefont {Lee}},\ and\ \bibinfo {author} {\bibfnamefont {J.-X.}\ \bibnamefont {Zhong}},\ }\href@noop {} {\bibinfo {title} {A fast and exact approach for stabilizer r{\'e}nyi entropy via the {XOR--FWHT} algorithm}} (\bibinfo {year} {2025}),\ \Eprint {https://arxiv.org/abs/2512.24685} {arXiv:2512.24685 [quant-ph]} \BibitemShut {NoStop}%
\bibitem [{\citenamefont {Xiao}\ and\ \citenamefont {Ryu}(2026)}]{Xiao2026FWHT}%
  \BibitemOpen
  \bibfield  {author} {\bibinfo {author} {\bibfnamefont {Z.}~\bibnamefont {Xiao}}\ and\ \bibinfo {author} {\bibfnamefont {S.}~\bibnamefont {Ryu}},\ }\bibfield  {title} {\bibinfo {title} {Exponentially accelerated sampling of pauli strings for nonstabilizerness},\ }\href {https://doi.org/10.1103/v3f2-ldr1} {\bibfield  {journal} {\bibinfo  {journal} {Physical Review Letters}\ }\textbf {\bibinfo {volume} {137}},\ \bibinfo {pages} {070402} (\bibinfo {year} {2026})},\ \Eprint {https://arxiv.org/abs/2601.00761} {2601.00761} \BibitemShut {NoStop}%
\bibitem [{\citenamefont {Sierant}\ \emph {et~al.}(2026)\citenamefont {Sierant}, \citenamefont {Vall{\`e}s-Muns},\ and\ \citenamefont {Garcia-Saez}}]{Sierant2026Magic}%
  \BibitemOpen
  \bibfield  {author} {\bibinfo {author} {\bibfnamefont {P.}~\bibnamefont {Sierant}}, \bibinfo {author} {\bibfnamefont {J.}~\bibnamefont {Vall{\`e}s-Muns}},\ and\ \bibinfo {author} {\bibfnamefont {A.}~\bibnamefont {Garcia-Saez}},\ }\bibfield  {title} {\bibinfo {title} {Computing quantum magic of state vectors},\ }\href {https://doi.org/10.22331/q-2026-04-10-2059} {\bibfield  {journal} {\bibinfo  {journal} {Quantum}\ }\textbf {\bibinfo {volume} {10}},\ \bibinfo {pages} {2059} (\bibinfo {year} {2026})},\ \Eprint {https://arxiv.org/abs/2601.07824} {2601.07824} \BibitemShut {NoStop}%
\bibitem [{\citenamefont {Fischer}\ \emph {et~al.}(2026)\citenamefont {Fischer}, \citenamefont {Leahy}, \citenamefont {Eddins}, \citenamefont {Keenan}, \citenamefont {Ferracin}, \citenamefont {Rossi}, \citenamefont {Kim}, \citenamefont {He}, \citenamefont {Pietracaprina}, \citenamefont {Sokolov}, \citenamefont {Dooley}, \citenamefont {Zimbor{\'a}s}, \citenamefont {Tacchino}, \citenamefont {Maniscalco}, \citenamefont {Goold}, \citenamefont {Garc{\'i}a-P{\'e}rez}, \citenamefont {Tavernelli}, \citenamefont {Kandala},\ and\ \citenamefont {Filippov}}]{Fischer2026Chaos}%
  \BibitemOpen
  \bibfield  {author} {\bibinfo {author} {\bibfnamefont {L.~E.}\ \bibnamefont {Fischer}}, \bibinfo {author} {\bibfnamefont {M.}~\bibnamefont {Leahy}}, \bibinfo {author} {\bibfnamefont {A.}~\bibnamefont {Eddins}}, \bibinfo {author} {\bibfnamefont {N.}~\bibnamefont {Keenan}}, \bibinfo {author} {\bibfnamefont {D.}~\bibnamefont {Ferracin}}, \bibinfo {author} {\bibfnamefont {M.~A.~C.}\ \bibnamefont {Rossi}}, \bibinfo {author} {\bibfnamefont {Y.}~\bibnamefont {Kim}}, \bibinfo {author} {\bibfnamefont {A.}~\bibnamefont {He}}, \bibinfo {author} {\bibfnamefont {F.}~\bibnamefont {Pietracaprina}}, \bibinfo {author} {\bibfnamefont {B.}~\bibnamefont {Sokolov}}, \bibinfo {author} {\bibfnamefont {S.}~\bibnamefont {Dooley}}, \bibinfo {author} {\bibfnamefont {Z.}~\bibnamefont {Zimbor{\'a}s}}, \bibinfo {author} {\bibfnamefont {F.}~\bibnamefont {Tacchino}}, \bibinfo {author} {\bibfnamefont {S.}~\bibnamefont {Maniscalco}}, \bibinfo {author} {\bibfnamefont {J.}~\bibnamefont {Goold}}, \bibinfo {author} {\bibfnamefont
  {G.}~\bibnamefont {Garc{\'i}a-P{\'e}rez}}, \bibinfo {author} {\bibfnamefont {I.}~\bibnamefont {Tavernelli}}, \bibinfo {author} {\bibfnamefont {A.}~\bibnamefont {Kandala}},\ and\ \bibinfo {author} {\bibfnamefont {S.~N.}\ \bibnamefont {Filippov}},\ }\bibfield  {title} {\bibinfo {title} {Dynamical simulations of many-body quantum chaos on a quantum computer},\ }\href {https://doi.org/10.1038/s41567-025-03144-9} {\bibfield  {journal} {\bibinfo  {journal} {Nature Physics}\ }\textbf {\bibinfo {volume} {22}},\ \bibinfo {pages} {302} (\bibinfo {year} {2026})}\BibitemShut {NoStop}%
\bibitem [{\citenamefont {Milne}(2020)}]{milneANT}%
  \BibitemOpen
  \bibfield  {author} {\bibinfo {author} {\bibfnamefont {J.~S.}\ \bibnamefont {Milne}},\ }\href {https://www.jmilne.org/math/CourseNotes/ANT.pdf} {\bibinfo {title} {Algebraic number theory}} (\bibinfo {year} {2020}),\ \bibinfo {note} {version 3.08}\BibitemShut {NoStop}%
\bibitem [{\citenamefont {Washington}(1997)}]{washington1997cyclotomic}%
  \BibitemOpen
  \bibfield  {author} {\bibinfo {author} {\bibfnamefont {L.~C.}\ \bibnamefont {Washington}},\ }\href {https://doi.org/10.1007/978-1-4612-1934-7} {\emph {\bibinfo {title} {Introduction to Cyclotomic Fields}}},\ \bibinfo {edition} {2nd}\ ed.,\ \bibinfo {series} {Graduate Texts in Mathematics}, Vol.~\bibinfo {volume} {83}\ (\bibinfo  {publisher} {Springer},\ \bibinfo {address} {New York},\ \bibinfo {year} {1997})\BibitemShut {NoStop}%
\bibitem [{\citenamefont {Keune}(2023)}]{Keune2023NumberFields}%
  \BibitemOpen
  \bibfield  {author} {\bibinfo {author} {\bibfnamefont {F.}~\bibnamefont {Keune}},\ }\href {https://doi.org/10.54195/IPVU4488} {\emph {\bibinfo {title} {Number Fields}}}\ (\bibinfo  {publisher} {Radboud University Press},\ \bibinfo {year} {2023})\BibitemShut {NoStop}%
\end{thebibliography}%

\appendix

\section{Background on Algebraic Number Theory}
\label{app:algebraic-number-theory-background}

This appendix reviews the standard concepts of algebraic integers, ideals, and ramification used in the main text.

\subsection{Algebraic Integers and Rings of Integers}
\label{app:ring-of-integers}

Let \(K\) be a number field, namely, a finite extension of \(\mathbb{Q}\).

\begin{definition}[Algebraic integer]
An element \(\alpha\in K\) is called an \emph{algebraic integer} if it satisfies a monic polynomial with integer coefficients. Equivalently, there exist \(a_0,\ldots,a_{r-1}\in\mathbb{Z}\) such that
\begin{equation}
    \alpha^r + a_{r-1}\alpha^{r-1} + \cdots + a_1 \alpha + a_0 = 0.
\end{equation}
\end{definition}

\begin{definition}[Ring of integers]
The set of all algebraic integers contained in \(K\) forms a ring, called the \emph{ring of integers} of \(K\). It is denoted by
\begin{equation}
    \mathcal{O}_K:=\left\{\alpha\in K:\alpha\text{ is an algebraic integer}\right\}.
\end{equation}
\end{definition}

\subsection*{Examples}

\begin{example}[The rational field]
For \(K=\mathbb{Q}\), the only rational numbers that satisfy a monic polynomial with integer coefficients are the ordinary integers. Therefore,
\begin{equation}
    \mathcal{O}_{\mathbb{Q}}=\mathbb{Z}.
\end{equation}
\end{example}

\begin{example}[The Gaussian field]
Consider
\begin{equation}
    K=\mathbb{Q}(i), \qquad i^2=-1.
\end{equation}
Since \(i\) satisfies the monic polynomial $x^2+1=0,$ it is an algebraic integer. The ring of integers of
\(\mathbb{Q}(i)\) is the ring of Gaussian integers,
\begin{equation}
    \mathcal{O}_{\mathbb{Q}(i)} =\mathbb{Z}[i]=\left\{a+bi:a,b\in\mathbb{Z}\right\}.
\end{equation}
\end{example}

\begin{example}[Cyclotomic fields]
Let  $ \zeta_n=e^{2\pi i/n}$ be a primitive \(n\)th root of unity, and consider the cyclotomic field $ K=\mathbb{Q}(\zeta_n).$ Because \(\zeta_n\) satisfies the monic polynomial $x^n-1=0,$ it is an algebraic integer. In fact, the complete ring of integers of the cyclotomic field is $ \mathcal{O}_{\mathbb{Q}(\zeta_n)}=\mathbb{Z}[\zeta_n] $  \cite[Theorem~6.4]{milneANT}
Equivalently,
\begin{equation}
    \mathbb{Z}[\zeta_n]=\left\{ \sum_{j=0}^{\varphi(n)-1}a_j\zeta_n^j: a_j\in\mathbb{Z} \right\},
\end{equation}
where \(\varphi(n)\) is Euler's totient function. The Gaussian field is
the special cyclotomic case $\mathbb{Q}(i)=\mathbb{Q}(\zeta_4).$
\end{example}

\subsection{Ideals and Prime Ideals}
\label{app:ideals-prime-ideals}

Throughout this section, \(A\) denotes a commutative ring with identity. In the applications considered in this work, \(A\) will usually be the ring of integers \(\mathcal{O}_K\) of a number field \(K\).

\begin{definition}[Ideal]
A nonempty subset \(I\subseteq A\) is called an \emph{ideal} of \(A\) if
\begin{equation}
    x,y\in I \quad\Longrightarrow\quad  x-y\in I,
\end{equation}
and
\begin{equation}
    a\in A,\quad x\in I \quad\Longrightarrow\quad ax\in I.
\end{equation}
Thus, an ideal is closed under addition and subtraction and absorbs multiplication by arbitrary elements of the ring.
\end{definition}

\begin{definition}[Principal ideal]
For an element \(\alpha\in A\), the \emph{principal ideal} generated by
\(\alpha\) is
\begin{equation}
    (\alpha):= \alpha A = \left\{\alpha a:a\in A \right\}.
\end{equation}
More generally, the ideal generated by
\(\alpha_1,\ldots,\alpha_r\in A\) is
\begin{equation}
    (\alpha_1,\ldots,\alpha_r) := \left\{ a_1\alpha_1+\cdots+a_r\alpha_r: a_1,\ldots,a_r\in A \right\}.
\end{equation}
\end{definition}

\begin{definition}[Prime ideal]
A proper ideal \(\mathfrak p\subsetneq A\) is called a \emph{prime ideal} if
\begin{equation}
    ab\in\mathfrak p  \quad\Longrightarrow\quad a\in\mathfrak p \quad\text{or}\quad b\in\mathfrak p
\end{equation}
for all \(a,b\in A\).
\end{definition}

\subsection*{Examples}

\begin{example}[Ideals in \(\mathbb{Z}\)]
Every ideal of \(\mathbb{Z}\) is principal and has the form
\begin{equation}
    (n) = n\mathbb{Z} = \left\{ nk:k\in\mathbb{Z} \right\}
\end{equation}
for some nonnegative integer \(n\). If \(p\) is a rational prime, then
\begin{equation}
    (p)=p\mathbb{Z}
\end{equation}
is a prime ideal of \(\mathbb{Z}\). For example, \(5\mathbb{Z}\) is a
prime ideal, whereas \(6\mathbb{Z}\) is not, since
\[
2\cdot3\in6\mathbb{Z}, \qquad 2\notin6\mathbb{Z}, \qquad 3\notin6\mathbb{Z}.
\]
\end{example}

\begin{example}[Prime ideals in the Gaussian integers]
For
\[
K=\mathbb{Q}(i), \qquad \mathcal{O}_K=\mathbb{Z}[i],
\]
consider the principal ideal
\begin{equation}
    \mathfrak p = (1+i) = \left\{ (1+i)(a+bi): a,b\in\mathbb{Z} \right\}.
\end{equation}
This is a prime ideal because
\begin{equation}
    \mathbb{Z}[i]/(1+i)
    \cong
    \mathbb{F}_2,
\end{equation}
which is a field. Moreover,
\begin{equation}
    (2)
    =
    (1+i)^2
    =
    \mathfrak p^2,
\end{equation}
where multiplication by a unit does not change a principal ideal. The
contraction of \(\mathfrak p\) to \(\mathbb{Z}\) is
\begin{equation}
    \mathfrak p\cap\mathbb{Z}
    =
    (2).
\end{equation}
Thus, \(\mathfrak p=(1+i)\) is a prime ideal of \(\mathbb{Z}[i]\) lying
above the rational prime \(2\).
\end{example}

\begin{example}[Prime ideals in a cyclotomic field]
Let
\[
K=\mathbb{Q}(\zeta_n),
\qquad
\mathcal{O}_K=\mathbb{Z}[\zeta_n].
\]
For a rational prime \(p\), the ideal generated by \(p\) in
\(\mathcal{O}_K\) is
\begin{equation}
    (p)
    =
    p\mathcal{O}_K.
\end{equation}
Since \(\mathcal{O}_K\) is a Dedekind domain, this ideal has a unique
factorization
\begin{equation}
    p\mathcal{O}_K
    =
    \mathfrak p_1^{\,e_1}
    \cdots
    \mathfrak p_g^{\,e_g}
\end{equation}
into distinct prime ideals of \(\mathcal{O}_K\).

As a concrete example, let \(\ell\) be a rational prime and consider
\[
K=\mathbb{Q}(\zeta_\ell).
\]
The ideal
\begin{equation}
    \mathfrak l
    =
    (1-\zeta_\ell)
\end{equation}
is a prime ideal of \(\mathbb{Z}[\zeta_\ell]\), and
\begin{equation}
    (\ell)
    =
    \mathfrak l^{\,\ell-1}
    =
    (1-\zeta_\ell)^{\ell-1}
    \cite[Theorem~6.4]{milneANT}.
\end{equation}
Its contraction to \(\mathbb{Z}\) is
\begin{equation}
    \mathfrak l\cap\mathbb{Z}
    =
    (\ell).
\end{equation}
Therefore, \(\mathfrak l=(1-\zeta_\ell)\) is the unique prime ideal
lying above \(\ell\) in \(\mathbb{Q}(\zeta_\ell)\).
\end{example}

\subsection{Ramification of Primes}

Let $K/\mathbb Q$ be a number field and $p$ a rational prime. The ideal generated by $p$ factors in $\mathcal O_K$ as
\begin{equation}
p\mathcal O_K=\prod_{j=1}^{g}\mathfrak p_j^{e_j},
\end{equation}
where the $\mathfrak p_j$ are distinct prime ideals. The exponents $e_j$ are the ramification indices; $p$ \emph{ramifies} if some $e_j>1$. In a quadratic extension, a rational prime has exactly three possible behaviors:
\begin{equation}
p\ \longrightarrow\ p\mathcal O_K=
\begin{cases}
\mathfrak p, & \text{inert},\\
\mathfrak p_1\mathfrak p_2,\quad \mathfrak p_1\ne\mathfrak p_2,
    & \text{split},\\
\mathfrak p^{\,2}, & \text{ramified}.
\end{cases}
\end{equation}

For example, let $K=\mathbb Q(i)$, so $\mathcal O_K=\mathbb Z[i]$. The primes $2$, $3$, and $5$ illustrate these three behaviors:
\begin{equation}
(2)=(1+i)^2,\qquad (3)\ \text{is prime},\qquad (5)=(2+i)(2-i).
\end{equation}
The ideal $(1+i)$ occurs twice, so $2$ ramifies. The ideal $(3)$ remains prime, so $3$ is inert. The two ideals $(2+i)$ and $(2-i)$ are distinct, so $5$ splits.

We can see the same behavior from the minimal polynomial  $f(x)=x^2+1$ of $i$ over $\mathbb{Q}$. Its reductions are
\begin{align}
f(x)\bmod 2 &= (x+1)^2,\\
f(x)\bmod 3 &= x^2+1
    \qquad\text{(irreducible over $\mathbb F_3$)},\\
f(x)\bmod 5 &= (x-2)(x+2).
\end{align}
The repeated factor, irreducible polynomial, and two distinct factors correspond to ramification, inertia, and splitting, respectively. This polynomial test gives the ideal factorizations directly here because $\mathcal O_K=\mathbb Z[i]$.
Finally, $\operatorname{Disc}(K)=\operatorname{Disc}(f)=-4$: its only prime divisor is $2$, precisely the prime that ramifies. 

The same notion applies to an extension $L/K$ of number fields, with rational primes replaced by nonzero prime ideals of $\mathcal O_K$. Let $\mathfrak p$ be such a prime ideal. In $\mathcal O_L$, it factors as
\begin{equation}
\label{eq:relative-prime-factorization}
\mathfrak p\mathcal O_L
=
\prod_{j=1}^{r}\mathfrak P_j^{e_j},
\end{equation}
where the $\mathfrak P_j$ are distinct prime ideals of $\mathcal O_L$ lying above $\mathfrak p$, meaning
$\mathfrak P_j\cap\mathcal O_K=\mathfrak p$.
Here, $\mathfrak p\mathcal O_L$ denotes the ideal of $\mathcal O_L$ generated by the elements of $\mathfrak p$.
We say that $\mathfrak p$ \emph{ramifies} in $L/K$ if some $e_j>1$.
Equivalently,
\begin{equation}
\mathfrak p\text{ ramifies in }L/K
\quad\Longleftrightarrow\quad
\mathfrak p\mid\mathfrak d_{L/K},
\end{equation}
where $\mathfrak d_{L/K}$ is the relative discriminant ideal in $\mathcal O_K$.

\section{Ramification in a tower of fields}\label{app:proof for ramification of primes}

\begin{proof}[Proof of Lemma~\ref{lem:ramification of primes}:]
Consider the tower $\mathbb{Q}\subseteq K\subseteq L.$ Choose a prime ideal $\mathfrak{P}\subseteq\mathcal{O}_L$ lying above \(p\) such that $e(\mathfrak{P}/p)>1.$
Define $\mathfrak{p} = \mathfrak{P}\cap\mathcal{O}_K.$

Then \(\mathfrak{p}\) is a prime ideal of \(\mathcal{O}_K\) lying above $p$. Hence, $\mathfrak{p}\cap\mathbb{Z}=(p).$ By the multiplicativity of ramification indices in a tower,
\[
e(\mathfrak{P}/p) = e(\mathfrak{P}/\mathfrak{p})\, e(\mathfrak{p}/p).
\]
Since $e(\mathfrak{P}/p)>1,$ at least one of the following two cases must occur.

\medskip

\noindent\textbf{Case A: $e(\mathfrak{p}/p)>1.$}

In this case, \(p\) is ramified in $K/\mathbb{Q}.$ 

\medskip

\noindent\textbf{Case B: $e(\mathfrak{P}/\mathfrak{p})>1.$}

In this case, the prime ideal \(\mathfrak{p}\) ramifies in $L/K.$
Equivalently, if
\[
\mathfrak{p}\mathcal{O}_L = \mathfrak{P}_1^{e_1} \mathfrak{P}_2^{e_2} \cdots \mathfrak{P}_g^{e_g},
\]
then $e_j>1$ for at least one \(j\).

\medskip

\noindent
\fbox{%
\parbox{0.92\linewidth}{%
If \(p\) ramifies in \(L/\mathbb{Q}\), then either \(p\) ramifies in
\(K/\mathbb{Q}\), or there exists a prime ideal
\(\mathfrak{p}\subseteq\mathcal{O}_K\) lying above \(p\) that ramifies
in \(L/K\).
}}

\end{proof}

\section{Prime Divisors of the Cyclotomic Discriminant}
\label{app:cyclotomic-discriminant}

Let $\zeta_n=e^{2\pi i/n}$
be a primitive \(n\)th root of unity, and let $K_n=\mathbb{Q}(\zeta_n)$ with $n>2.$
The degree of \(K_n\) over \(\mathbb{Q}\) is $[K_n:\mathbb{Q}]=\varphi(n),$ where \(\varphi(n)\) is Euler's totient function.
The discriminant of the cyclotomic field is
\begin{equation}
\label{eq:cyclotomic-discriminant}
    \operatorname{Disc}(K_n/\mathbb{Q}) = (-1)^{\varphi(n)/2} \frac{n^{\varphi(n)}} {\displaystyle \prod_{p\mid n} p^{\varphi(n)/(p-1)} },
\end{equation}
as given in Ref.~\cite[Proposition~2.7]{washington1997cyclotomic}.

Let \(q\) be a rational prime and let \(v_q(n)\) denote the largest nonnegative integer \(a\) such that \(q^a\mid n\). If \(q\mid n\), Eq.~\eqref{eq:cyclotomic-discriminant} gives
\begin{equation}
\label{eq:cyclotomic-discriminant-valuation}
    v_q\!\left( \left| \operatorname{Disc}(K_n/\mathbb{Q}) \right|\right) =\varphi(n)  \left( v_q(n)-\frac{1}{q-1} \right).
\end{equation}

For every odd prime \(q\mid n\), the right-hand side is strictly
positive. Hence,
\begin{equation}
    q\mid n,\quad q\text{ odd} \quad\Longrightarrow\quad q\mid\operatorname{Disc}(K_n/\mathbb{Q}).
\end{equation}
For \(q=2\), Eq.~\eqref{eq:cyclotomic-discriminant-valuation} becomes
\begin{equation}
    v_2\!\left(  \left| \operatorname{Disc}(K_n/\mathbb{Q})  \right| \right) = \varphi(n)\bigl(v_2(n)-1\bigr).
\end{equation}
Therefore,
\begin{equation}
    2\mid\operatorname{Disc}(K_n/\mathbb{Q}) \quad\Longleftrightarrow\quad 4\mid n.
\end{equation}

Consequently, for any rational prime \(q\),
\begin{equation}
\label{eq:prime-divisors-cyclotomic-discriminant}
\boxed{
\begin{aligned}
q\mid\operatorname{Disc}(K_n/\mathbb{Q}) \quad\Longleftrightarrow\quad
&\bigl(q\mid n\text{ and }q\text{ is odd}\bigr)\\
&\text{or }\bigl(q=2\text{ and }4\mid n\bigr).
\end{aligned}
}
\end{equation}

In particular, if $ n\not\equiv2\pmod4,$ then the rational prime divisors of
\(\operatorname{Disc}(\mathbb{Q}(\zeta_n)/\mathbb{Q})\) are precisely the rational prime divisors of \(n\). Since ramified rational primes
are exactly the prime divisors of the field discriminant \cite[Theorem~3.35]{milneANT}, the same statement holds for the ramified primes.

When \(n\equiv2\pmod4\), every odd prime divisor of \(n\) still divides the discriminant, but \(2\) does not. Therefore, in the recurrence
algorithm, the prime \(2\) is included separately.

\section{Proof of the Ramification--Discriminant Lemma}
\label{app:ramification_discriminant_proof}

\begin{proof}[Proof of Lemma~\ref{lem:ramification-discriminant-cleared}:]
Let $f_{\mathrm{sf}}\in K[x]$ be monic and square-free, and let $L$ be its splitting field over $K$. Choose
$c\in\mathbb Z_{>0}$ such that
\begin{equation}
F_{\mathrm{sf}}(x)=c f_{\mathrm{sf}}(x) \in\mathcal O_K[x].
\end{equation}
The leading coefficient of $F_{\mathrm{sf}}$ is $c$ and its splitting field is also $L$.

We use the standard polynomial reduction criterion: if a polynomial in $\mathcal O_K[x]$ has leading coefficient not divisible by a prime ideal $\mathfrak p$, and its reduction modulo $\mathfrak p$ has no repeated roots, then $\mathfrak p$ is unramified in its splitting field \cite[Sec.~7.8, p.~180]{Keune2023NumberFields}. If $p\mid c$, the first alternative holds.

Suppose therefore that $p\nmid c$. Since $\mathfrak p$ lies above $p$, we have $\mathfrak p\cap\mathbb Z=(p)$. Thus $c\notin\mathfrak p$, so the leading coefficient of $F_{\mathrm{sf}}$ remains nonzero modulo $\mathfrak p$. In particular, reduction modulo $\mathfrak p$ preserves its degree.
Write
\begin{equation}
\Delta=\operatorname{Disc}(F_{\mathrm{sf}}) \in\mathcal O_K.
\end{equation}

Suppose that $\mathfrak p$ does not divide $(\Delta)$. Then $\Delta$ remains nonzero modulo $\mathfrak p$. Because the discriminant is a polynomial expression in the coefficients, reduction commutes with taking
the discriminant:
\begin{equation}
\operatorname{Disc}\!\left(\overline{F_{\mathrm{sf}}}\right) = \overline{\operatorname{Disc}(F_{\mathrm{sf}})} = \overline{\Delta} \neq 0.
\end{equation}
Here the bar denotes reduction modulo $\mathfrak p$. Consequently, $\overline{F_{\mathrm{sf}}}$ has no repeated roots in an algebraic closure of
$\mathcal O_K/\mathfrak p$.

The polynomial reduction criterion therefore implies that $\mathfrak p$ is unramified in $L/K$. This contradicts the hypothesis. Hence $\mathfrak p\mid(\Delta).$ Taking ideal norms preserves this divisibility. Since $\mathfrak p$ lies above $p$, its ideal norm is $p^f$ for some integer $f\geq1$. Therefore
\begin{equation}
p^f = \operatorname{Nm}(\mathfrak p) \mid \operatorname{Nm}((\Delta)) = \left|N_{K/\mathbb Q}(\Delta)\right|.
\end{equation}
It follows that
\begin{equation}
p\mid
\left|N_{K/\mathbb Q}
\bigl(\operatorname{Disc}(F_{\mathrm{sf}})\bigr)\right|.
\end{equation}
This proves the second alternative when $p\nmid c$,
and completes the proof.
\end{proof}

\section{Ramification Argument for Finding Prime Divisors of the Period} \label{app:prrof for the theorem of the prime divisors of the recurrence period}

\begin{proof}[Proof of theorem \ref{thm:prime-divisors-of-cyclotomic-order}:] 

Let $U$ be a unitary matrix whose entries lie in a number field $K$, and let $n$ be the smallest positive integer such that $U^n=\tau I$ for some scalar $\tau$. Let $L$ be the splitting field of $f_{\mathrm{sf}}$ over $K$. Using our previous work, \cite{anand2026arithmetic} we establish that $\mathbb Q(\zeta_n)\subseteq L.$ Now let $p$ be an odd prime dividing $n$. By Appendix~\ref{app:cyclotomic-discriminant},
$p$ ramifies in $\mathbb Q(\zeta_n)/\mathbb Q$. Since ramification indices multiply in the tower $\mathbb Q\subseteq\mathbb Q(\zeta_n)\subseteq L,$
the prime $p$ also ramifies in $L/\mathbb Q$.

Applying Lemma~\ref{lem:ramification of primes} to the tower $\mathbb Q\subseteq K\subseteq L$ gives two possibilities.
If $p$ ramifies in $K/\mathbb Q$, then
\begin{equation}
p\mid\operatorname{Disc}(K/\mathbb Q),
\end{equation}
by the field-discriminant criterion, see Appendix \ref{app:cyclotomic-discriminant}.

Otherwise, some prime ideal $\mathfrak p$ of $\mathcal O_K$ above $p$ ramifies in $L/K$. The ramification--discriminant lemma proved in
Appendix~\ref{app:ramification_discriminant_proof}
then gives
\begin{equation}
p\mid c
\qquad\text{or}\qquad
p\mid
\left|N_{K/\mathbb Q}
\bigl(\operatorname{Disc}(F_{\mathrm{sf}})\bigr)\right|.
\end{equation}
This proves the assertion for every odd prime divisor of $n$. Including $2$ separately completes the proof.
\end{proof}

To obtain the possible prime divisors of $n$, define
the positive integer
\begin{equation}
M= 2c\, \left|\operatorname{Disc}(K/\mathbb Q)\right| \left|N_{K/\mathbb Q} \bigl(\operatorname{Disc}(F_{\mathrm{sf}})\bigr)\right|.
\end{equation}
The square-freeness of $f_{\mathrm{sf}}$ ensures that
$\operatorname{Disc}(F_{\mathrm{sf}})\neq0$, so $M\neq0$.
Factor $M$ as
\begin{equation}
M=\prod_{i=1}^{s}p_i^{b_i},
\qquad b_i\geq1.
\end{equation}
The theorem implies that
\begin{equation}
\boxed{
n=\prod_{i=1}^{s}p_i^{a_i},
\qquad a_i\in\mathbb Z_{\geq0}.
}
\end{equation}
Some of the exponents $a_i$ may vanish: the primes $p_i$ form a set containing all prime divisors of $n$, but need not all divide $n$. The exponents $a_i$ are bounded separately using the cyclotomic degree bound.

\section{Upper Bound on the Exponents}\label{app:proof for finding maximum exponent}

The prime divisors of $n$ have been constrained by ramification. We now bound their exponents using the degree of the splitting field. Let
$d=\deg f_{\mathrm{sf}}$. Since $\operatorname{Gal}(L/K)$ acts faithfully on the $d$ distinct roots of $f_{\mathrm{sf}}$, we have $[L:K]\mid d!$.
Together with $\mathbb Q(\zeta_n)\subseteq L$, this gives
\begin{equation}
\label{eq:degree-bound-phi}
\varphi(n) \mid [K:\mathbb{Q}]\,d!.
\end{equation}

Suppose that the prime factorization of $n$ is $n=\prod_i p_i^{a_i}$, where each $a_i\geq 1$. Euler's totient formula
gives
\begin{equation}
\label{eq:euler-phi-factorization}
\varphi(n)=\prod_i p_i^{a_i-1}(p_i-1).
\end{equation}
For each prime divisor $p_i$ of $n$, it follows that
\begin{equation}
\label{eq:valuation-bound}
a_i-1 \leq v_{p_i}\!\left(\varphi(n)\right) \leq v_{p_i}\!\left([K:\mathbb Q]\,d!\right).
\end{equation}
Therefore, the exponent $a_i$ satisfies
\begin{equation}
\label{eq:exponent-upper-bound}
a_i\leq 1+v_{p_i}\!\left([K:\mathbb Q]\,d!\right).
\end{equation}

\begin{example}
Let $n=24=2^3\cdot3$. Then $\varphi(24)=8$, so $v_2(\varphi(24))=3$ and $v_3(\varphi(24))=0$. Consequently, the degree condition requires
$8\mid[K:\mathbb Q]\,d!$.
\end{example}

\section{Proof of the polynomial-remainder criterion}
\label{app:polynomial-remainder-proof}

\begin{proof}[Proof of Lemma~\ref{thm:polynomial-remainder}:]

By the defining property of the minimal polynomial, for every polynomial
\(q(x)\in K[x]\), \[ q(U)=0 \quad\Longleftrightarrow\quad f_{\mathrm{sf}}(x)\mid q(x). \]
Applying this result to $ q(x)=x^k-\tau$ gives
\begin{align}
    U^k=\tau I
    &\quad\Longleftrightarrow\quad
    U^k-\tau I=0
    \nonumber\\
    &\quad\Longleftrightarrow\quad
    f_{\mathrm{sf}}(x)\mid x^k-\tau.
\end{align}

By the polynomial division algorithm, there exists a polynomial
\(a_k(x)\in K[x]\) such that
\[
x^k=a_k(x)f_{\mathrm{sf}}(x)+\rho_k(x), \qquad \deg\rho_k<d_{\mathrm{sf}}. 
\]
The divisibility condition
\[
f_{\mathrm{sf}}(x)\mid x^k-\tau
\]
holds if and only if the remainder of \(x^k-\tau\) upon division by
\(f_{\mathrm{sf}}(x)\) is zero. Since this remainder is $\rho_k(x)-\tau,$
the divisibility condition is equivalent to $\rho_k(x)=\tau.$
Therefore,
\[
U^k=\tau I \quad\Longleftrightarrow\quad \rho_k(x)=\tau,
\]
as claimed.
\end{proof}

\section{Analytical proof of the recurrent non-Clifford family}
\label{app:analytical_recurrence_proof}
Here we provide the analytical proof for exact recurrence in $N-$qubit kicked Ising open chain. We take
\begin{equation}
h=\frac{\pi}{2}, \qquad J=\frac{\pi}{2s}, \qquad s\in\mathbb{Z}_{>0}.
\end{equation}
The transverse-field unitary is
\begin{equation}
U_h = \exp\!\left(-ih\sum_{j=1}^{L}\sigma_y^{(j)}\right).
\end{equation}
Here, $\sigma_\alpha^{(j)}$ denotes the Pauli operator $\sigma_\alpha$, with $\alpha=x,y,z$, acting on site $j$. We use the convention
\begin{equation}
e^{-i\theta\hat{\mathbf n}\cdot\boldsymbol{\sigma}}
= \cos\theta\,I_2 -i\sin\theta\,\hat{\mathbf n}\cdot\boldsymbol{\sigma},
\end{equation}
where $\hat{\mathbf n}$ is a unit vector. At $h=\pi/2$, this becomes
\begin{equation}
U_h=(-i)^N\prod_{j=1}^{N}\sigma_y^{(j)}.
\end{equation}
 Define
\begin{equation}
\label{eq:HZZ_HZ_definitions}
\begin{aligned}
H_{ZZ} &= \sum_{j=1}^{N-1} \sigma_z^{(j)}\sigma_z^{(j+1)},\\
H_Z &= g\sum_{j=1}^{N}\sigma_z^{(j)} +b_1\sigma_z^{(1)} +b_N\sigma_z^{(L)},
\end{aligned}
\end{equation}
so that $H_{\mathrm{Ising}}=JH_{ZZ}+H_Z$. The $H_{ZZ}$ and $H_Z$ label the interaction and
longitudinal-field terms, respectively.
Conjugation by $U_h$ acts on the single-site Pauli operators as
\begin{equation}
\begin{aligned}
U_h\sigma_x^{(j)}U_h^\dagger&=-\sigma_x^{(j)},\quad
U_h\sigma_y^{(j)}U_h^\dagger&=\sigma_y^{(j)}, \quad \\
 \qquad U_h\sigma_z^{(j)}U_h^\dagger&=-\sigma_z^{(j)}.
\end{aligned}
\end{equation}
Consequently,
\begin{equation}
\label{eq:global_Y_properties}
U_h^2=(-1)^N I_{2^N},
\quad
U_hH_{ZZ}U_h^\dagger=H_{ZZ},
\quad
U_hH_ZU_h^\dagger=-H_Z.
\end{equation}
Using the Floquet convention
\begin{equation}
U_F = U_h\exp\!\left[-i\left(JH_{ZZ}+H_Z\right)\right],
\end{equation}
and noting that $[H_{ZZ},H_Z]=0$, we obtain
\begin{equation}
\begin{aligned}
U_F^2 &= U_h e^{-i(JH_{ZZ}+H_Z)} U_h e^{-i(JH_{ZZ}+H_Z)}\\
&= (-1)^N e^{-i(JH_{ZZ}-H_Z)} e^{-i(JH_{ZZ}+H_Z)}\\
&= (-1)^N e^{-i2JH_{ZZ}}.
\end{aligned}
\label{eq:two_step_recurrent_family}
\end{equation}
Thus, the longitudinal and boundary fields cancel exactly after two Floquet periods. Substituting $J=\pi/(2s)$ and raising Eq.~\eqref{eq:two_step_recurrent_family} to the $ s-$th power gives
\begin{equation}
\begin{aligned}
U_F^{2s}
&= (-1)^{Ns}e^{-i\pi H_{ZZ}}\\
&= (-1)^{Ns} \prod_{j=1}^{N-1} e^{-i\pi\sigma_z^{(j)}\sigma_z^{(j+1)}}\\
&= (-1)^{Ns+N-1}I_{2^N}.
\end{aligned}
\label{eq:exact_nonclifford_recurrence}
\end{equation}
This is an exact recurrence. To see that $2s$ is its minimal period for $N\geq2$, first observe that an odd power of $U_F$ contains the global spin flip and cannot be proportional to the identity. For an even power $U_F^{2q}$ to be proportional to the identity, eigenvalues of $H_{ZZ}$ that differ by two must acquire the same phase. Hence $e^{-i4qJ} = e^{-i2\pi q/s} = 1,$
which first holds for $q=s$. Therefore, $n_{\mathrm{rec}}=2s, $ for $ N\geq2.$

\end{document}